\documentclass[lettersize,journal]{IEEEtran}
\usepackage{amsmath,amsfonts}
\usepackage{algorithm}
\usepackage{array}
\usepackage[caption=false,font=normalsize,labelfont=sf,textfont=sf]{subfig}
\usepackage{textcomp}
\usepackage{stfloats}
\usepackage{url}
\usepackage{verbatim}
\usepackage{graphicx}
\graphicspath{{figures/}}
\usepackage{cite}
\usepackage{hyperref}
\hypersetup{
    colorlinks=true,
    linkcolor=blue,
    citecolor=blue,
    urlcolor=black
}
\usepackage{amssymb} 
\usepackage{algorithmicx}
\usepackage{algpseudocode}
\usepackage{amsthm}  
\newtheorem{lemma}{Lemma}
\newtheorem{theorem}{Theorem}
\usepackage{multirow} 
\usepackage{threeparttable}
\usepackage{booktabs}
\begin{document}

\title{FOBNN: Fast Oblivious Inference via Binarized Neural Networks}

\author{Xin Chen, Zhili Chen, Shiwen Wei, Junqing Gong, and Lin Chen 
\thanks{This work has been submitted to the IEEE for possible publication. Copyright may be transferred without notice, after which this version may no longer be accessible.}
\thanks{Xin Chen, Zhili Chen, Shiwen Wei, and Junqing Gong are with the Software Engineering Institute, East China Normal University, Shanghai 200062, China (e-mail: xinchen@stu.ecnu.edu.cn; zhlchen@sei.ecnu.edu.cn; 51265902068@stu.ecnu.edu.cn; jqgong@sei.ecnu.edu.cn).}
\thanks{Lin Chen is with the Engineering Research Centre of Applied Technology on Machine Translation and Artificial Intelligence, Macao Polytechnic University, Macau, SAR, China (e-mail: lchen@mpu.edu.mo).}}

\markboth{Journal of \LaTeX\ Class Files,~Vol.~14, No.~8, August~2021}%
{Shell \MakeLowercase{\textit{et al.}}: A Sample Article Using IEEEtran.cls for IEEE Journals}


\maketitle

\begin{abstract}
The remarkable performance of deep learning has sparked the rise of Deep Learning as a Service (DLaaS), allowing clients to send their personal data to service providers for model predictions. A persistent challenge in this context is safeguarding the privacy of clients' sensitive data. Oblivious inference allows the execution of neural networks on client inputs without revealing either the inputs or the outcomes to the service providers. In this paper, we propose FOBNN, a Fast Oblivious inference framework via Binarized Neural Networks. In FOBNN, through neural network binarization, we convert linear operations (e.g., convolutional and fully-connected operations) into eXclusive NORs (XNORs) and an Oblivious Bit Count (OBC) problem. For secure multiparty computation techniques, like garbled circuits or bitwise secret sharing, XNOR operations incur no communication cost, making the OBC problem the primary bottleneck for linear operations. To tackle this, we first propose the Bit Length Bounding (BLB) algorithm, which minimizes bit representation to decrease redundant computations. Subsequently, we develop the Layer-wise Bit Accumulation (LBA) algorithm, utilizing pure bit operations layer by layer to further boost performance. We also enhance the binarized neural network structure through link optimization and structure exploration. The former optimizes link connections given a network structure, while the latter explores optimal network structures under same secure computation costs. Our theoretical analysis reveals that the BLB algorithm outperforms the state-of-the-art OBC algorithm by a range of 17\% to 55\%, while the LBA exhibits an improvement of nearly 100\%. Comprehensive proof-of-concept evaluation demonstrates that FOBNN outperforms prior art on popular benchmarks and shows effectiveness in emerging bioinformatics.
\end{abstract}


\section{Introduction}
\IEEEPARstart{T}{he} growing trend of Deep Learning as a Service (DLaaS) is particularly noteworthy. In this model, service providers host deep learning models on their servers. Clients then submit their data to these providers for analysis, in exchange for valuable insights. While DLaaS enables the deployment of deep neural networks on resource-constrained devices, it simultaneously raises significant privacy concerns due to the potential exposure of sensitive client information to these service providers.

To address the privacy challenge, oblivious inference has attracted significant attention. This technique allows clients to utilize models hosted by service providers without revealing the privacy of their inputs or outputs, while simultaneously protecting the models for the service providers. Various cryptographic protocols have been proposed for oblivious inference, including homomorphic encryption (HE) \cite{liu2017oblivious,juvekar2018gazelle}, secret sharing (SS) \cite{ibarrondo2021banners,rathee2020cryptflow2}, garbled circuits (GCs) \cite{rouhani2018deepsecure,riazi2019xonn} and their hybrids \cite{mohassel2017secureml,zhang2023scalable}. 
Each of these cryptographic tools possesses unique characteristics and involves various trade-offs. HE enables computations on encrypted data but it is computaionally intensive and its complexity scales with the circuits' depth. A standalone SS-based scheme is computationally efficient, yet requires three or more non-colluding computation parties, which is a strong assumption. Additionally, the communication rounds count scales linearly with the circuits' depth involved.  

On the other hand, GC supports arbitrary functionalities while requiring only a constant round of communications disregarding the circuits' depth. Nevertheless, it incurs a high communication cost and a significant overhead for both multiplication and addition, the fundamental operations in deep learning. Hybrid-protocols aim to harness the strengths of various protocols, but their round complexity scales linearly with the depth of the deep learning model, which is crucial for deep learning accuracy. In the context of oblivious inference, constant communication rounds are favored as clients prefer initial or final interaction with service providers, rather than continuous engagement during the inference process. Thus, the focus of this paper is GC-based oblivious inference.

Typically, GC-based oblivious inference for traditional neural networks incurs substantial performance overheads. The reason is that these neural networks rely heavily on linear operations consisting of multiplications and additions, and both computations in GCs are computationally intensive. To tackle this, we binarize the traditional neural networks and apply GCs to the binarized neural networks (BNNs) \cite{courbariaux2016binarized}. BNNs excel at achieving efficient model compression, thereby accelerating computations \cite{wu2016quantized,hubara2018quantized}. Furthermore, through binarization, we can convert multiplications into 
nearly cost-free XNORs \cite{riazi2019xonn} and formulate additions as an Oblivious Bit Count (OBC) problem, which involves the oblivious accumulation of numerous secret bits. As a result, GCs for BNNs are greatly accelerated, and the OBC problem emerges as the new performance bottleneck for linear operations.

From the above, we can identify two major challenges in designing GC-based oblivious inference for BNNs. The first challenge lies in optimizing linear operations by improving the solution to the OBC problem. The state-of-the-art solution is the tree-adder \cite{riazi2019xonn}, but it suffers from non-compact bitwise representations as illustrated in Section \ref{sec:blb-motivation}. How to design a more efficient algorithm for the OBC problem is critical to overcome this challenge. The second challenge involves making up for the accuracy loss due to binarization by optimizing neural network structures while factoring in secure inference costs. Binarization must cause loss of information and thus loss of accuracy. How to minimize the accuracy loss and oblivious inference costs by adjusting network structures is challenging.

To address these two challenges, we introduce FOBNN, a fast oblivious inference framework based on customized BNNs. Fundamentally, FOBNN customizes a BNN for a given traditional neural network to maintain its accuracy as much as possible. Subsequently, it enhances the efficiency of linear operations by devising efficient solutions to the OBC problem. Particularly, in FOBNN, we devise the Bit Length Bounding (BLB) algorithm, which provides optimal bit representations based on the range analysis of values, and thus significantly reduces redundant computations. BLB outperforms the tree-adder \cite{riazi2019xonn} by a range of 17\% to 55\%. To further enhance efficiency, we develop the Layer-wise Bit Accumulation (LBA) algorithm, leveraging layer-wise bit operations to minimize performance overheads. LBA achieves an improvement of nearly 100\% upon the tree-adder.
Ultimately, FOBNN optimizes network structures of BNNs to improve accuracy by link optimization and structure exploration. The former optimizes link connections under a given neural network structure by training ternary weights including zero values and deleting links with zero-weights, while the latter explores optimal network structures under same secure computation costs.
    
We highlight our contributions as follows.

\begin{itemize}
    \item \textbf{FOBNN Framework.} We propose FOBNN, a fast GC-based oblivious inference framework via BNNs. FOBNN effectively overcomes the substantial performance overhead of GCs for traditional neural networks, while retains its advantage of constant rounds of communications.
    \item \textbf{Linear Computation Acceleration.} We develop two fast oblivious algorithms for accelerating linear operations, Bit Length Bounding and Layer-wise Bit Accumulation. Both algorithms minimize the number of non-XOR gates by utilizing compact bitwise representations, and thus are of independent interest for secure computation techniques such as garbled circuits and bitwise secret sharing. 
    \item \textbf{Network Structure Optimization.} We devise efficient methods to optimize link connections under a given network structure and explore optimal network structures under same secure computation costs to improve the accuracy of oblivious inference. This effectively makes up for the accuracy loss due to neural network binarization.
    \item \textbf{FOBNN Implementation and Application.} We implement FOBNN, conduct extensive experiments and demonstrate its application in popular classification and bioinformatics benchmarks. Experimental results show that our scheme outperforms state-of-the-art methods. 
\end{itemize}

The remainder of this paper is structured as follows. Section \ref{sec:Preliminaries} introduces technical preliminaries. Section \ref{sec:Sqni} details the design of our FOBNN framework. Section \ref{sec:BLB} and Section \ref{sec:LBA} outlines BLB and LBA algorithms, respectively, and provides theoretical complexity analysis. Section \ref{sec:BNSCO} discusses BNN structure optimization. In Section \ref{sec:Experiments}, comprehensive evaluations of FOBNN framework are conducted. Section \ref{sec:RW} briefly reviews related work, and finally, Section \ref{sec:Conclusion} concludes this paper.

\section{Preliminaries} \label{sec:Preliminaries}
\subsection{Neural Networks}
A Convolutional Neural Network (CNN) is composed of multiple layers that perform different functions. The output of each layer serves as the input of the next layer. Below we describe the functionality of different layers.

\emph{Linear Layers}. The linear layers in a CNN contain Convolutional (CONV) layers and Fully-Connected (FC) layers. The vector dot product (VDP) operations are performed to compute outputs in these layers. 

\emph{Batch Normalization Layers}. Batch normalization (BN) is primarily utilized to standardize features and accelerate network convergence. A BN layer typically follows a linear layer, normalizing its output. The trainable parameters of BN layers are learned during the training phase \cite{ioffe2015batch}. Concisely, it can be expressed as $Y=\gamma'X+\beta'$, where $\gamma'$ is positive and both $\gamma'$ and $\beta'$ can be inferred from the trained parameters \cite{zhu2022securebinn}.

\emph{Activation Layers}. Activation functions apply non-linear transformations while preserving input-output dimensionality. $Sigmod$ and Rectified Linear Unit ($ReLU$) are two commonly used activation functions in CNNs.

\emph{Pooling Layers}. The pooling layers are used to reduce the size of outputs. Popular pooling methods include Max-pooling (MP) and Average-pooling (AP). We adopt MP in this paper. 

We adopt the idea of BNNs \cite{courbariaux2016binarized} to quantize CNNs. Specifically, we replace $Sigmod$ or $ReLU$ with $Sign$ activation function, and the weights and activations in our BNNs are constrained to $+1$ or $-1$.

\subsection{Ternary Weight Networks}\label{sec:twn}
Ternary weight networks (TWNs) are a subtype of neural networks with weights constrained to $+\alpha$, 0 and $-\alpha$ \cite{li2022ternary}. The simple, efficient, and accurate TWNs can address the limited storage and computational resources issues for real world AI applications. TWNs seek to strike a balance between the BNNs counterparts and the full precision weight networks (FPWNs) counterparts. The ternary weights can be viewed as sparse binary weights, thus we substitute the binary weights of BNNs with ternary weights in our optimization scheme to achieve better network performance and efficient storage and computation. In the inference phase, we extract $\alpha$ and convert the weights into ternary values: $+1$, 0 and $-1$. The $\alpha$ can be transferred to the BN layers for computation, and the computation of the CONV layers is aligned with the binary-valued computation of the BNNs.

\subsection{Knowledge Distillation}
Knowledge distillation (KD), a representative type of model compression and acceleration, can effectively learn a small student model from a large teacher model \cite{gou2021knowledge}. In KD, a small student model is generally supervised by a large teacher model, and the student model mimics the teacher model to obtain competitive or even superior performance \cite{hinton2015distilling}. KD's ability to improve model performance on resource-limited devices is an excellent fit with our BNNs, hence we use it to compensate for the loss of accuracy after binarization.

\subsection{Garbled Circuits}
Garbled circuits (i.e., Yao's protocol), or GC in short, is a generic cryptographic tool for secure two-party computation \cite{yao1986generate,lindell2009proof}. In GC, an arbitrary function $f$ can be computed based on private inputs held by two parties, without revealing any information about each party's input to the other. Specifically, one holds its private input $x_1$, and the other holds $x_2$. First, one party (i.e., the garbler) prepares a GC to evaluate $f(x_1,x_2)$ and transmits the GC, garbled labels on input $x_1$, and the output decoding table to the other party (called the evaluator). The evaluator needs to run a 1-out-of-2 oblivious transfer (OT) protocol with the garbler \cite{rabin2005exchange} to get its garbled labels on input ${x_2}$ obliviously. Then, the evaluator computes the garbled output with garbled labels of two parties. Finally, the plain output can be decoded by the output decoding table.

\subsection{Threat Model}
In oblivious inference, a client retain its private input data $x$ while a service provider possesses trained deep-learning model $\theta$. They collaborate through a secure evaluation protocol to perform neural network inference on the client's input $x$ using the service provider's model $\theta$. The client obtain the inference result without disclosing its input $x$ or outcome to the service provider, while the service provider's model $\theta$ remain concealed from the client.

We adhere to the honest-but-curious adversary model, as established in prior works \cite{juvekar2018gazelle, riazi2019xonn, zhang2023scalable}. Here, neither party fully trusts the other, but both follow the protocol's specifications. Both aim to extract information from exchanged messages regarding the other's private data.

\section{The FOBNN Framework} \label{sec:Sqni}
The FOBNN achieves fast GC-based oblivious inference via BNNs. Figure \ref{fig:architecture} illustrates the workflow of FOBNN. Specifically, in FOBNN, we first convert CNNs into customized BNNs by training them while simultaneously binarizing their weights, where KD is employed to maximize accuracy retention. Then the BNNs are converted into boolean circuits, where we accelerate linear operations by proposing two efficient solutions to the OBC problem: Bit Length Bounding and Layer-wise Bit Accumulation. In addition, we also optimize BNN structures by link optimization and structure exploration to obtain more efficient boolean circuits. Finally, the resulted boolean circuits are employed in the garbled circuit protocol to achieve oblivious inference.

\begin{figure*}[htbp]
    \centering
    \begin{minipage}{0.64\textwidth}
        \centering
        \includegraphics[width=\textwidth]{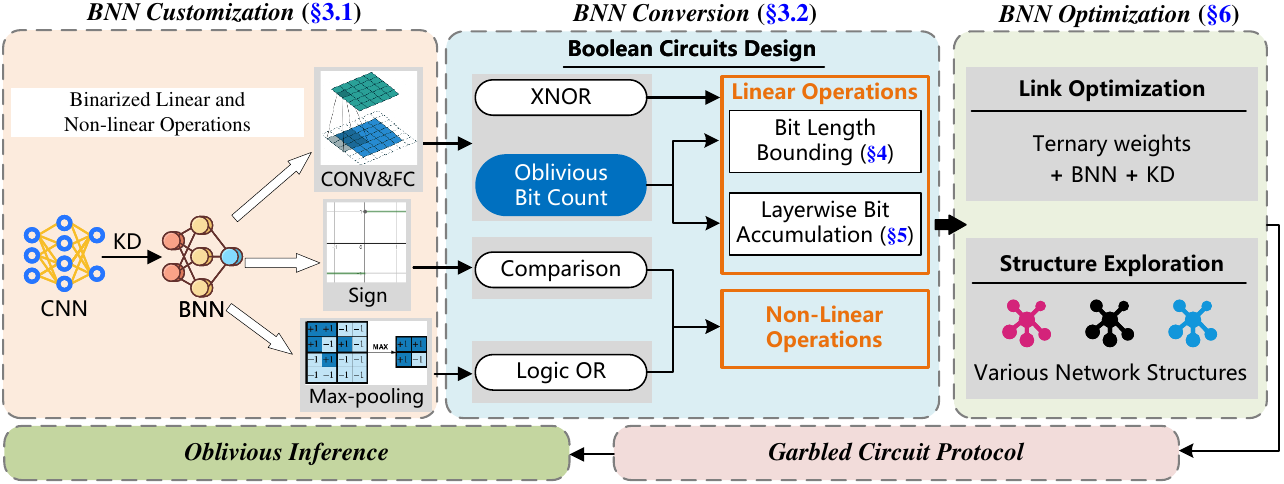}
        \caption{The Workflow of FOBNN Framework.}
        \label{fig:architecture}
    \end{minipage}
    \hfill 
    \begin{minipage}{0.34\textwidth}
        \centering
        \includegraphics[width=\textwidth]{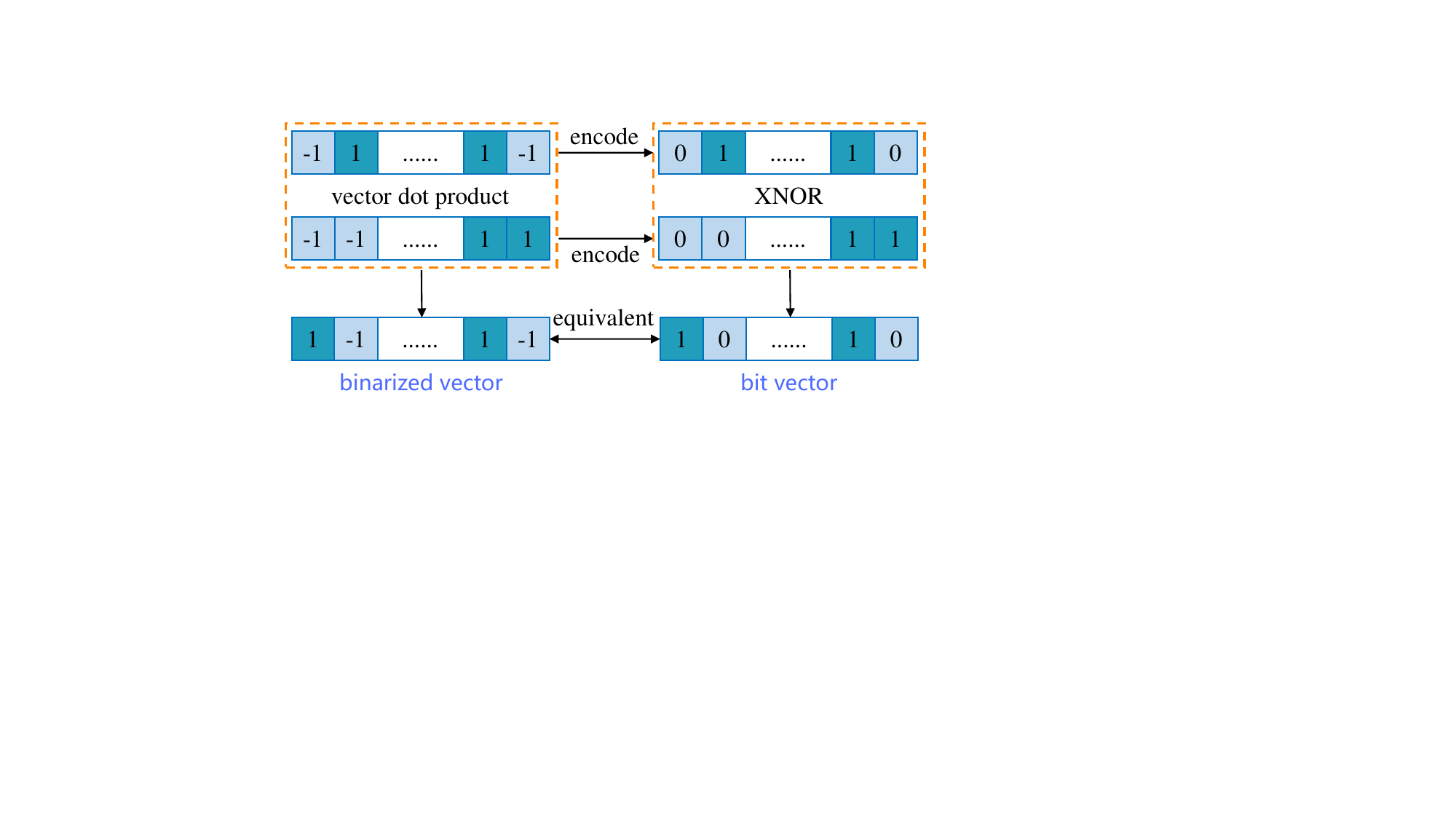}
        \caption{Encode and XNOR operation of two binary vectors.}
        \label{fig:xnor}
    \end{minipage}
\end{figure*}

In this section, we first convert CNNs into customized BNNs in Section \ref{sec:CQN}. Then we introduce BNN conversion to boolean circuits and BNN optimization in Section \ref{sec:OI}. Section \ref{sec:SA} discusses the security of our FOBNN framework.

\subsection{BNN Customization}\label{sec:CQN}
A notable limitation of BNNs is their accuracy. To tackle it, we develop CNN-based BNNs via tailored structure and parameter optimization. It preserves the baseline CNN's structure with minimal depth/kernel modifications and accuracy retention. 
Furthermore, we employ the teacher-student structure of KD to enhance accuracy. 
 
\emph{Customized CONV.} CONV layer weights are binarized ($\pm1$) without requiring bias vectors. Customization focuses on kernel size, filter count, padding, and strides. We maintain the original kernel size for 2D convolutions or appropriately increase it for 1D convolutions, while tuning filter count and other parameters.

\emph{Customized Activation.} We substitute the conventional activation function (such as $ReLU$) with a binary activation (BA) function, specifically the $Sign$ function \cite{courbariaux2016binarized}. Additionally, we incorporate a BN layer followed by the BA function. 

\emph{Customized Pooling.} MP is preferable to AP in BNNs because the maximum of binary elements retains its binary nature \cite{riazi2019xonn}. Therefore, binary MP is employed to downscale the output. The key adjustable parameter is the pool size, with strides typically set equal to the pool size by default.

\emph{Customized FC.} FC layers, also known as dense layers, are commonly positioned before the output layer. The weights and custom parameters in these layers are manipulated in a manner analogous to those in the customized CONV layers.

On the whole, during network binarization, we carefully increase network depth and layer parameters to preserve original accuracy with minimal structural addition. Unlike prior methods involving complex operations like pruning \cite{riazi2019xonn} or slimmable networks \cite{zhang2023scalable}, our customization avoids these additional operations.

\subsection{BNN Conversion and Optimization}\label{sec:OI}
Each layer of BNNs can be computed by using a simplified arithmetic circuit. Now, we describe how to convert different BNN layers into their corresponding boolean circuits.

\textbf{Oblivious Linear Layers:}  In our binarized networks, both inputs and weights of CONV and FC layers are binary. Linear operations are performed using VDP, and their corresponding boolean circuits can be realized through XnorPopcount operations \cite{courbariaux2016binarized}. Here, we reformulate these operations as XNOR operations combined with the OBC problem. The process is outlined below.

(1) Convert Multiplications to XNORs: We examine the VDP between an input vector $\mathbf{x} \in \{-1,+1\}^{n}$ and a weight vector $\mathbf{w} \in \{-1,+1\}^{n}$. By encoding $-1$ as $0$ and $+1$ as $1$, we observe that the multiplication operation exhibits symmetry and can be reduced to an eXclusive NOR (XNOR) operation \cite{riazi2019xonn}, as illustrated in Figure \ref{fig:xnor}.

(2) Convert Additions to OBC problem: After performing XNOR operations on the converted bit vectors, it suffices to obliviously count the number $1$'s in the resulted bit vector. Let $c_1$ denote the number of $1$'s, and $L_v$ the size of the vector. Consequently, the number of $0$'s (or $-1$ equivalents) is $L_v - c_1$. The output $X$ of the binarized vector is then calculated as 
\begin{equation}\label{eq:vdp-output}
    X = c_1-(L_v - c_1)=2 \cdot c_1 - L_v.
\end{equation}

Since XNORs are nearly overhead-free in garbled circuits, the VDP operations heavily used in neural networks are converted into the OBC problem. We propose two fast solutions to this problem in Section \ref{sec:BLB} and Section \ref{sec:LBA}, which are of independent interest for secure computation techniques such as GCs and bitwise SS. 

\textbf{Oblivious BA Function Layers:} We employ the $Sign(\cdot)$ function as the BA function to output $+1$ or $-1$. In BNNs, a BN layer precedes an activation function layer \cite{zhu2022securebinn}. The combination of these two layers together with Eq.~\eqref{eq:vdp-output} can be realized through a comparison circuit. Specifically, since $\gamma'$ is positive, the combination can be simplified as: 
\begin{equation}\label{eq:bn-sign}
    \begin{aligned}
        Y&=Sign(\gamma'X+\beta')=Sign(\gamma'(2 \cdot c_1 - L_v)+\beta')\\
        &=Sign(2\gamma' c_1 - \gamma' L_v + \beta')\\
        &=Sign(c_1 - \frac{\gamma' L_v - \beta'}{2\gamma'}).
    \end{aligned}
\end{equation}
Thus, it can be realized by comparing $c_1$ and $\frac{\gamma' L_v - \beta'}{2\gamma'}$. We formally optimize the computations of OBC, BN, and activation across three layers for the first time. 

\textbf{Oblivious Binary MP Layers:} We prefer MP to AP in this layer. Since the inputs of MP are binarized, MP is equivalent to performing logical OR over the binary encodings as depicted in \cite{riazi2019xonn}.

To enhance the efficiency of boolean circuits for oblivious inference, we can optimize BNN structures in two ways. Firstly, we refine the link connections within a given BNN structure. Secondly, we explore various BNN structures while maintaining consistent secure computation costs. The BNN structure optimizations are detailed in Section \ref{sec:BNSCO}.

\subsection{Security of FOBNN}\label{sec:SA}
The FOBNN framework leverages GCs to protect against semi-honest adversaries, in line with state-of-the-art oblivious inference models \cite{mohassel2017secureml, liu2017oblivious, rouhani2018deepsecure, riazi2019xonn}. Participants follow the protocol but seek extra information beyond the outcome from communications and internal states. 

In FOBNN framework, a service provider converts a CNN to a customized BNN, converts the BNN to a single boolean circuit, and may also optimize the BNN structure for a more efficient boolean circuit. The service provider holds this boolean circuit as input, converts it into a GC and sends the GC to a client. The client holds a binarized and encoded sample as input, garbles the input into labels with the help of the service privider using an Oblivious Transfer (OT) protocol, and then computes the garbled circuit with the garbled input to get a garbled output, which can be decoded by the decoding table provided by the service provider. This forms a traditional GC protocol. Therefore, the security of FOBNN framework can be reduced to the security of GC protocol, which has been theoretically established in \cite{lindell2009proof}.

\section{Bit Length Bounding}\label{sec:BLB}
In this section, we present our first efficient solution to the OBC problem. Firstly, we describe the motivation. Subsequently, we devise the Bit Length Bounding algorithm. Finally, we analyze the complexity of the BLB algorithm, and compare it with the state-of-the-art method, the tree-adder (TA).

\subsection{Motivation}\label{sec:blb-motivation}
As we know, the most critical step for linear operations of BNNs is to solve the OBC problem. The state-of-the-art solution is the tree-adder \cite{riazi2019xonn}. Specifically, the tree-adder operates by summing two 1-bit numbers to yield a 2-bit result, and extends this process to sum two 2-bit numbers to produce a 3-bit output, and so forth. Our observation is that the sum of any two 1-bit numbers is at most $2$, which is less than the maximum value representable by a $2$-bit number, $(11)_2 = 3$. Similarly, the sum of any two $2$-bit numbers does not exceed $6$, which is again less than the maximum $3$-bit number, $(111)_2 = 7$. More generally, the sum of any two $n$-bit numbers is bounded by $2^{n+1} - 2$, which is inherently less than the maximum $(n+1)$-bit number, $2^{n+1}-1$. Furthermore, this limitation in bit representation can accumulate when multiple sums are involved. For example, the tree-adder sums two $1$-bit numbers to a $2$-bit number at most $2$, and then it sums two $2$-bit numbers resulted previously to at most $4$. The bit representation insufficiency is enlarged.

\subsection{BLB Algorithm}
We propose Bit Length Bounding algorithm to address the above issue. The core idea is to deduce, based on the desired output, the number of inputs to sum and their maximum bitwise representation. Starting with $2$-bit numbers, we note that their maximum value is $3$, which can be expressed as $3=2^2-1=(2^1+1)(2^1-1)=3 \times 1$. Since $1$ is the maximum $1$-bit number, summing three $1$-bit numbers yields the maximum $2$-bit number. Next, we consider summing $2$-bit numbers to form a $3$-bit number. The maximum $2$-bit value is $3$, and the maximum $3$-bit value is $7$. Since $7$ is not divisible by $3$, multiple maximum $2$-bit numbers cannot sum to the maximum $3$-bit number. However, multiple $2$-bit numbers can sum to the maximum $4$-bit number, which is $15$. Specifically, $2^4-1=(2^2+1)(2^2-1)=5 \cdot (2^2-1)$, indicating that summing $5$ maximum $2$-bit numbers yields the maximum $4$-bit number. Generalizing this pattern, when inputs are $2^i$-bit numbers ($i \geq 1$), we sum $2^{2^i}+1$ such numbers to produce a $2^{i+1}$-bit number. This is because the maximum values of $2^i$-bit and $2^{i+1}$-bit numbers are $2^{2^i}-1$ and $2^{2^{i+1}}-1$, respectively, and $2^{2^{i+1}}-1=(2^{2^i}+1)(2^{2^i}-1)$. In summary, our algorithm leverages the properties of bitwise representation to determine the number of inputs required to produce a desired output.

Our BLB algorithm can be divided into two steps as follows.

\subsubsection{Bit Computation} 
In this step, we transform $1$-bit numbers into $2$-bit numbers. Given that two bits can represent values up to $3$, we group every three adjacent $1$-bit numbers. We sum each group using a specialized $1$-bit adder \cite{kolesnikov2009improved} as shown in Figure \ref{fig:adder}. This adder accepts three input bits $a$, $b$, and $c$, and produces two output bits: the sum bit $d_0 = a \oplus b \oplus c$ and the carry-out bit $d_1 = c \oplus((a \oplus c) \wedge (b \oplus c))$. It is very efficient for GCs since it uses only one non-XOR gate. This process results in a set of 2-bit numbers corresponding to each input bit vector.

\begin{figure}[htbp] 
	\centering
	\includegraphics[width=0.24\textwidth]{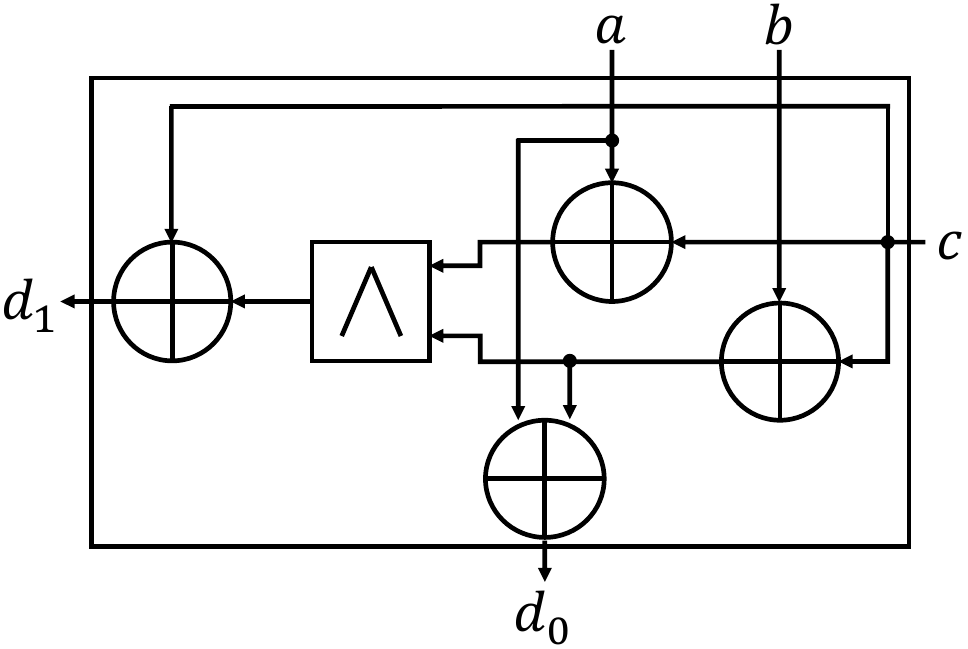}
	\caption{1-bit Adder.}
	\label{fig:adder}
\end{figure}

\subsubsection{Cross-layer Computation} 
Once we have the set of $2$-bit numbers, we can iteratively perform cross-layer computations as detailed below. Assuming the inputs to the cross-layer computation are $2^p$-bit numbers ($p \in \mathbb{Z}^+$), as discussed earlier, the output of this computation comprises $2^{p+1}$ bits.  We begin by organizing the inputs into groups, with each group containing $2^{2^p}+1$ elements. Within each group, we compute the sum of the first $2^{2^p}$ inputs using a tree adder outlined in Algorithm \ref{alg:algorithm1}, proceeding level by level until we reach a single $2^{p+1}$-bit sum. Subsequently, we add the remaining $2^p$-bit number to this sum to obtain the output for each group. By bounding the bit length of inputs and outputs, we ensure that the computation does not overflow, and the output remains a $2^{p+1}$-bit number. Algorithm \ref{alg:algorithm2} provides a detailed description of this process.

Figure \ref{fig:2to4bit} demonstrates our BLB algorithm. We start by dividing a target bit vector of fifteen $1$-bit numbers (grey squares) into five groups of three. A $1$-bit adder circuit processes each group to produce a $2$-bit number (blue squares), comprising a sum and a carry-out bit. Next, in the Cross-layer Computation phase, the first four of these $2$-bit numbers are summed using the tree adder. This involves grouping them into pairs and summing to produce two $3$-bit numbers (purple squares) and then summing these to get a $4$-bit number (green squares), alongside a remaining $2$-bit number. The results are $9$ ($4$ bits) and $1$ ($2$ bits). Their sum does not exceed the maximum unsigned $4$-bit value ($15$), and the final result is a $4$-bit number representing $10$ (i.e., $(1010)_2=10$).

\begin{algorithm}[H]
    \caption{Tree Adder}
    \label{alg:algorithm1}
    \begin{algorithmic}[1]
        \Require $2^p$-bit numbers set $\mathbb{X}=\{x_i\}_{i \in [1..l]}$, $x_i \in \{0,1\}^p$     
        \Ensure A $q$-bit number $S_q$, where $q > 2^p$  
        \State {$q \gets 2^p, \mathbb{X}^0 \gets \mathbb{X}$}
        \While{$|\mathbb{X}^0|>1$}
            \State {$\mathbb{X}^1 \gets \varnothing$, $i \gets 1$, $group \gets |\mathbb{X}^0|/2$, $r \gets |\mathbb{X}^0| \bmod 2$}
            \For{$j=1$ to $group$}
                \State {$s_j \gets \mathbb{X}^0[i]+\mathbb{X}^0[i+1]$}  \Comment{standard GC addition}
                \State {$\mathbb{X}^1 \gets \mathbb{X}^1 \cup \{s_j\}$, $i \gets i+2$}
            \EndFor
            \If{$r=1$} 
                \State {$\mathbb{X}^1 \gets \mathbb{X}^1 \cup \mathbb{X}^0[i]$}
            \EndIf
            \State {$\mathbb{X}^0 \gets \mathbb{X}^1$, $q \gets q+1$}
        \EndWhile 
        \State {$S_q \gets \mathbb{X}^0[0]$}
    \end{algorithmic}
\end{algorithm}

\begin{algorithm}[H]
    \caption{Cross-layer Adder}
    \label{alg:algorithm2}
    \begin{algorithmic}[1]
        \Require $2^p$-bit numbers set $\mathbb{X}=\{x_i\}_{i \in [1..n]}$, $x_i \in \{0,1\}^{2^p}$    
        \Ensure $2^{p+1}$-bit numbers set $\mathbb{Y}=\{y_j\}_{j \in [1..n^{\prime}]}$, where $n^{\prime}=\lceil \frac{n}{2^{2^p}+1} \rceil$, $y_j \in \{0,1\}^{2^{p+1}}$ or a $q$-bit number $S_q$ 
        \State {$t \gets 1, ct \gets \lceil \frac{n}{2^{2^p}+1} \rceil$}  \Comment{count of groups}
        
        \If{$n<2^{2^p}+1$}  \Comment{not enough to form a group}
            \State {$S_q \gets Tree(\mathbb{X})$}  \Comment{sum by tree via Algorithm \ref{alg:algorithm1}} 
            \State {\textbf{return} $S_q$} 
        \EndIf
        
        \For{$j=1$ to $ct$}  \Comment{compute by groups}
            \State {$\mathbb{X}^j \gets \{ \mathbb{X}[t],\cdots,\mathbb{X}[t+2^{2^p}-1] \}$}  \Comment{first $2^{2^p}$ values}
            \State {$S_{2^{p+1}} \gets Tree(\mathbb{X}^j)$}  \Comment{sum by tree via Algorithm \ref{alg:algorithm1}}
            \State {$\mathbb{Y}[j] \gets S_{2^{p+1}} + \mathbb{X}[t+2^{2^p}]$, $t \gets t+2^{2^p}+1$} 
        \EndFor \\
        \Return {$\mathbb{Y}$}
    \end{algorithmic}
\end{algorithm}

\begin{figure}[htbp] 
	\centering
	\includegraphics[width=0.40\textwidth]{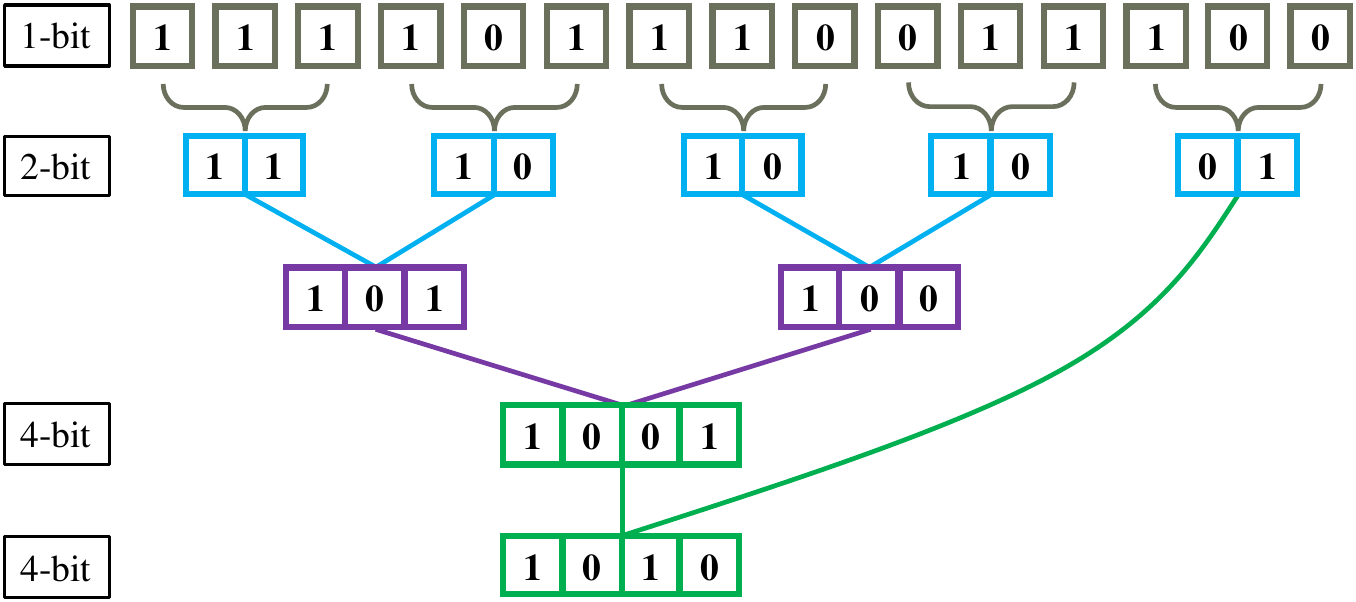}
	\caption{A toy example of BLB solution.}
	\label{fig:2to4bit}
\end{figure}

\subsection{Theoretical Complexity Analysis}\label{sec:blb-xor}
Since XOR operations are well known to be nearly overhead-free in GC protocol \cite{kolesnikov2008improved}, we analyze the complexity of the BLB algorithm by counting the number of non-XOR gates it needs. We compare the complexity of BLB algorithm with that of the tree-adder \cite{riazi2019xonn}, and quantify the improvement of BLB algorithm upon the tree adder. We denote the length of a bit vector by $N$, and the maximum bit length that can represent the result by $L=\lceil \log_{2}(N+1) \rceil$.

\begin{lemma}\label{lem:ts}
    The number of non-XOR gates that the tree adder needs is ${S}_{N}^{ts} = 2\cdot(N-1)-\log_{2}N \approx 2N $.
\end{lemma}

\begin{proof}
    The $\ell$-bit adders ($\ell \in [1..L]$) are used to compute the final result in the TA structure. The number of non-XOR gates for each $\ell$-bit adder is $\ell$, and the number of $\ell$-bit adders is $\frac{N}{2^{\ell}}$. Thus, the total number of non-XOR gates can be computed as follows. 
    \begin{equation}\label{lemma-eq3}
        {S}_{N}^{ts} = \sum_{\ell=1}^{L}(\frac{N}{2^{\ell}} \cdot \ell) = 1 \cdot \frac{N}{2} + 2 \cdot \frac{N}{2^2} + 3 \cdot \frac{N}{2^3} + \cdots + L \cdot \frac{N}{2^L}
    \end{equation}

    Multiplying both sides of Eq.~\eqref{lemma-eq3} by $\frac{1}{2}$ yields Eq.~\eqref{lemma-eq4}. 
    \begin{equation}\label{lemma-eq4}
    	\begin{aligned}
    		\frac{1}{2} {S}_{N}^{ts} = 1 \cdot \frac{N}{2^2} + 2 \cdot \frac{N}{2^3} + 3 \cdot \frac{N}{2^4} + \cdots + L \cdot \frac{N}{2^{L+1}}
    	\end{aligned}
    \end{equation}
    
    Eq.~\eqref{lemma-eq3} minus Eq.~\eqref{lemma-eq4}:
    \begin{equation}\label{lemma-eq5}
    	\begin{aligned}
    		{S}_{N}^{ts} - \frac{1}{2} {S}_{N}^{ts} = \frac{N}{2} + \frac{N}{2^2} + \frac{N}{2^3} + \cdots + \frac{N}{2^L} - L \cdot \frac{N}{2^{L+1}}
    	\end{aligned}
    \end{equation}
    
    We hold $L=\log_{2}{N}$. Eq.~\eqref{lemma-eq5} can be simplified as follows.
    \begin{equation}
    	\begin{aligned}
    		\frac{1}{2} {S}_{N}^{ts} &= \frac{N}{2} + \frac{N}{2^2} + \frac{N}{2^3} + \cdots + \frac{N}{2^{\log_{2}{N}}} - \frac{\log_{2}{N} \cdot N}{2^{\log_{2}{N}+1}} \\
    		{S}_{N}^{ts} &= N \cdot (1+\frac{1}{2}+\frac{1}{2^2}+ \cdots + \frac{1}{2^{\log_{2}{N}-1}}) - \log_{2}N \\
    		{S}_{N}^{ts} &= N \cdot 2 \cdot \frac{N-1}{N} - \log_{2}N \\
    		{S}_{N}^{ts} &=2 \cdot (N-1) - \log_{2}N
    	\end{aligned}
    	\nonumber
    \end{equation}

    The proof is completed.
\end{proof}

The complexity of BLB algorithm is shown in Theorem \ref{theo1}.

\begin{theorem}\label{theo1}
    The number of non-XOR gates that BLB algorithm needs is $${S}_{K}^{blb} = \sum_{\kappa=1}^{K-1}(\frac{N}{\prod_{\ell=0}^{\kappa}(2^{2^{\ell}}+1)} \cdot (2^{2^{\kappa}} \cdot 2^{\kappa} + 2^{2^{\kappa}} - 2)) + \frac{N}{3},$$ and $1.29N < {S}_{K}^{blb} < 1.71N$, where $K =\lceil \log_{2}\log_{2}(N+1) \rceil$ and $N>255$.
\end{theorem}

\begin{proof}
    Our BLB algorithm is divided into two steps, and the number of non-XOR gates is also divided into two parts that are counted separately. The number of non-XOR gates in the {\itshape Bit Computation} step is $\frac{N}{3}$. In the $\kappa$-th {\itshape Cross-layer Computation}, the amount of groups is $A_{\kappa} = \frac{N}{\prod_{\ell=0}^{\kappa}(2^{2^{\ell}}+1)}$ 
    and the number of non-XOR gates in each group can be computed by $G_{\kappa} = 2^{2^{\kappa}} \cdot 2^{\kappa} + 2^{2^{\kappa}} - 2$. 
    Appendix \ref{app:2} shows the derivation process of $G_{\kappa}$. 
    Thus, the total number of non-XOR gates of BLB algorithm can be expressed as follows. 
    \begin{equation}
        \begin{aligned}
            {S}_{K}^{blb} 
            =& \sum_{\kappa=1}^{K-1}(A_{\kappa} \cdot G_{\kappa}) + \frac{N}{3} \\
            =&\sum_{\kappa=1}^{K-1}(\frac{N}{\prod_{\ell=0}^{\kappa}(2^{2^{\ell}}+1)} \cdot (2^{2^{\kappa}} \cdot 2^{\kappa} + 2^{2^{\kappa}} - 2)) + \frac{N}{3} \\
            <&\sum_{\kappa=1}^{K-1}(\frac{N \cdot (2^{2^{\kappa}}+1) \cdot (2^{\kappa}+1)}{\prod_{\ell=0}^{\kappa-1}(2^{2^{\ell}}+1) \cdot (2^{2^{\kappa}}+1)} ) + \frac{N}{3} \\
            <&\frac{N}{3} + \frac{3N}{3} + \frac{5N}{3 \cdot 5} + \frac{9N}{3 \cdot 5 \cdot 17} + \frac{17N}{3 \cdot 5 \cdot 17 \cdot 257} + \cdots \\
            <&\frac{5N}{3} + \frac{N}{3 \cdot 5 \cdot 17} \cdot (9+1) < 1.71N
        \end{aligned}
        \nonumber
    \end{equation}
    It is worth noting that ${S}_{K}^{blb} < 1.71N$ holds for all $N>0$.
    Similarly, 
    \begin{equation}
        \begin{aligned}
            {S}_{K}^{blb} 
            >&\sum_{\kappa=1}^{K-1}(\frac{N \cdot (2^{2^{\kappa}}+1) \cdot 2^{\kappa}}{\prod_{\ell=0}^{\kappa-1}(2^{2^{\ell}}+1) \cdot (2^{2^{\kappa}}+1)}) + \frac{N}{3} \\
            =&\frac{N}{3} + \frac{2N}{3} + \frac{4N}{3 \cdot 5} + \frac{8N}{3 \cdot 5 \cdot 17} + \frac{16N}{3 \cdot 5 \cdot 17 \cdot 257} + \cdots \\
        \end{aligned}
        \nonumber
    \end{equation}

    if $N>255$ holds, we have ${S}_{K}^{blb} > 1.29N$.
\end{proof}

Consequently, we can conclude that our BLB algorithm outperforms the tree adder by about $(2N-1.71N)/(1.71N) = 17\%$ to $(2N-1.29N)/(1.29N) = 55\%$.

\section{Layer-wise Bit Accumulation} \label{sec:LBA}
In this section, we discuss the drawback of BLB algorithm and propose a more efficient algorithm called Layer-wise Bit Accumulation. Besides, we analyze the complexity of LBA algorithm and compare it with that of BLB and TA algorithms.

\subsection{Motivation}
While our BLB algorithm achieves compactness in its final bit representation and efficiency in its bit computations, it remains less efficient during cross-layer computations. For example, the bit computations from $1$-bit numbers to $2$-bit numbers improve the bit representation by $(3-2)/2=50\%$ compared to the tree adder, while the cross-layer computations from $2$-bit numbers to $4$-bit numbers improve only by $(15-14)/14 = 7\%$. Therefore, the question arises: is it feasible to extend bit computations universally to all computations?

\subsection{LBA Algorithm}
The answer is yes and this results in our LBA algorithm. Our design rationale is to perform bit accumulation layer by layer in the bitwise representation. Specifically, we hold a set of $1$-bit numbers for each bit vector. Each bit represents a value of $2^0$-layer. We firstly compute the output of the first layer (i.e., the $2^0$-layer) as follows. Following the method described by {\itshape Bit Computation}, we sum every three numbers in the set by the adder as shown in Figure \ref{fig:adder} to get a carry-out bit and a sum bit, which represent a value of $2^1$-layer and $2^0$-layer, respectively. Then the carry-out bit is added to the set of the $2^1$-layer, and the sum bit is added to the set of the $2^0$-layer. We perform the above computation continuously until there is one element left in the set of the $2^0$-layer. 
Thus, we get a set of bits for the $2^1$-layer, and the output bit of the $2^0$-layer. We then take the set of the $2^1$-layer as input and utilize a similar way to the computation process of the $2^0$-layer to produce a set of bits for the $2^2$-layer and the output bit of the $2^1$-layer. Further, the set of the $2^2$-layer is used as an input to the next process and is computed iteratively until all bits of the bitwise representations are computed. 

\begin{algorithm}[H]
    \caption{Layer-wise Bit Accumulation}
    \label{alg:algorithm3}
    \begin{algorithmic}[1]
        \Require bit vector $\mathbb{X}=\{x_i\}_{i \in [1..n]}$, $x_i \in \{0,1\}$     
        \Ensure $q$ bits vector $\mathbb{Y}=\{y_k\}_{k \in [1..q]}$, $y_k \in \{0,1\}$  
        \State {$\mathbb{X}^C \gets \mathbb{X}, q \gets \lceil \log_2(n+1) \rceil$}
        \For{$k=1$ to $q$}  \Comment{layer-wise computation}
            \State {$\mathbb{X}^0 \gets \varnothing, \mathbb{X}^1 \gets \varnothing, gp \gets |\mathbb{X}^C|/3, r \gets |\mathbb{X}^C|\bmod3$}  
            \While{$gp \geq 1$}
                \State {$t \gets 1$}
                \For{$j=1$ to $gp$}  \Comment{compute by groups}
                    \State {$d_1,d_0 \gets adder(\mathbb{X}^C[t],\mathbb{X}^C[t+1],\mathbb{X}^C[t+2])$} 
                    \State {$\mathbb{X}^1 \gets \mathbb{X}^1 \cup \{d_1\}, \mathbb{X}^0 \gets \mathbb{X}^0 \cup \{d_0\}, t \gets t+3$}
                \EndFor
                \If{$r > 0$} 
                    \State {$\mathbb{X}^0 \gets \mathbb{X}^0 \cup \{\text{Last $r$ elements of } \mathbb{X}^C\}$}
                \EndIf
                \State {$\mathbb{X}^C \gets \mathbb{X}^0, \mathbb{X}^0 \gets \varnothing$} 
                \State {$gp \gets |\mathbb{X}^C|/3, r \gets |\mathbb{X}^C|\bmod3$}
            \EndWhile
            \If{$r = 1$}  \Comment{$\mathbb{X}^C=\{x_0\}$}
                \State {$y_k \gets \mathbb{X}^C[0]$}
            \ElsIf{$r = 2$}  \Comment{$\mathbb{X}^C=\{x_0,x_1\}$}
                \State {$d_1,d_0 \gets adder(\mathbb{X}^C[0],\mathbb{X}^C[1],0)$}
                \State {$\mathbb{X}^1 \gets \mathbb{X}^1 \cup \{d_1\}, y_k \gets d_0$}
            \EndIf
            \State {$\mathbb{X}^C \gets \mathbb{X}^1$} 
        \EndFor
            
    \end{algorithmic}
\end{algorithm}

Algorithm \ref{alg:algorithm3} describes our LBA algorithm. The input of LBA is a bit vector and the output is a vector of $q$ bits, where $q$ denotes the minimum bits that can represent the result of the target computation. The algorithm processes the bits layer by layer, accumulating their values. Below, we detail the computations for each layer.

We perform bit computations by 1-bit adder in the bit vector $\mathbb{X}^C$ and define two sets $\mathbb{X}^1$ and $\mathbb{X}^0$ to store carry-out bits and sum bits, respectively. We divide every three elements of $\mathbb{X}^C$ into a group. As long as there are groups left, we compute iteratively as follows.

\begin{itemize}
	\item[1)] For each group, we invoke the 1-bit adder to obtain a carry-out bit $d_1$ and a sum bit $d_0$ (line 7).
	\item[2)] The carry-out bit $d_1$ and the sum bit $d_0$ responding to each group is added to the sets $\mathbb{X}^1$ and $\mathbb{X}^0$, respectively (line 8).
	\item[3)] All elements in the last group with less than three elements are added to the set $\mathbb{X}^0$ (lines 10-12).
	\item[4)] Assign the set $\mathbb{X}^0$ to the set $\mathbb{X}^C$ and update the relevant parameters of the set $\mathbb{X}^C$ (line 13-14).
\end{itemize}

After performing the above operations, the number of elements in the set $\mathbb{X}^C$ will be either one or two (lines 16-21). If there is only one element in the set (i.e., $\mathbb{X}^C=\{x_0\}$), it means that the value of the currently lowest bit (i.e., $y_k$) is $x_0$. Otherwise, the carry-out bit $d_1$ and the sum bit $d_0$ are computed by 1-bit adder, taking as input the two elements along with a value of 0. $d_1$ is added to $\mathbb{X}^1$, and the currently lowest bit of the output is $d_0$. At this point, $\mathbb{X}^1$ is assigned to $\mathbb{X}^C$ and all elements in $\mathbb{X}^1$ become the values of the currently lowest bits. The above computation process is repeated to obtain $y_1, y_2, \cdots, y_q$. The output vector $\mathbb{Y}=\{y_k\}_{k \in [1..q]}$ is actually the bit representation of the number of $1$'s $c_1=\sum_{k=1}^{q}y_{k} \cdot 2^{k-1}$.  

Figure \ref{fig:eg_ac} demonstrates the computation process of a bit vector consisting of nine elements. First, we divide the nine 1-bit values into three groups. Each group is computed using a 1-bit adder circuit depicted in Figure \ref{fig:adder} to obtain a carry-out bit and a sum bit, which represents a value of layers $2^1$ and $2^0$, respectively. The three computed sum bits form a group and the adder is invoked again to compute a carry-out bit and a sum bit. Thus, we get four 1-bit values of layer $2^1$ and a final value of layer $2^0$. Second, we divide the four 1-bit values of layer $2^1$ into groups. The first three elements form a group and are computed to obtain a carry-out bit and a sum bit, which represent values of layers $2^2$ and $2^1$, respectively. To compute a final value of layer $2^1$, the adder is invoked again, where the inputs are the remaining element and the sum bit along with an auxiliary 1-bit element 0. We hold two values of layer $2^2$ and a final value of layer $2^1$ in this step. Then, we put the two values of layer $2^2$ and an auxiliary bit 0 into a group. The final value of layer $2^3$ and the final value of layer $2^2$ (i.e., the carry-out bit and the sum bit) can be computed by the 1-bit adder. Finally, we hold four final bits, and actually obtain the bit representation of the number of $1$'s $c_1 = 0 \cdot 2^0 + 1 \cdot 2^1 + 1 \cdot 2^2 + 0 \cdot 2^3 = 6$.  

\begin{figure}[ht] 
	\centering
	\includegraphics[width=0.32\textwidth]{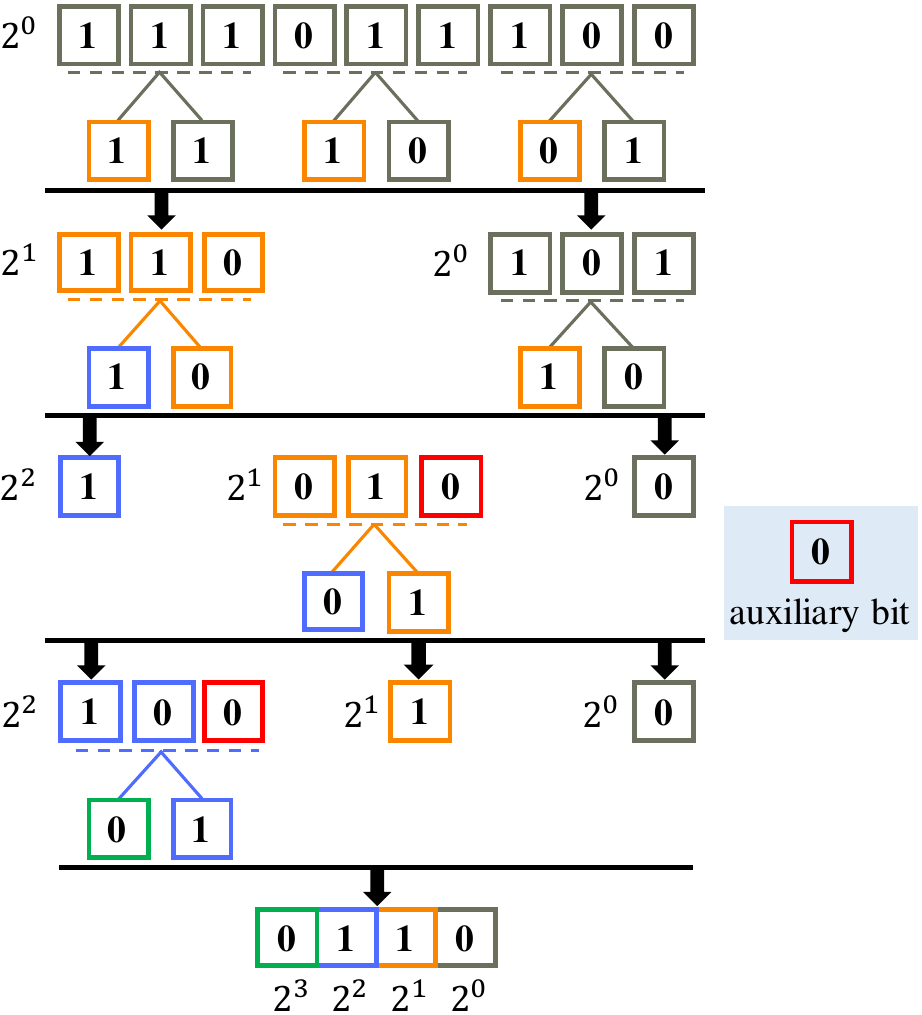}
	\caption{An example of LBA algorithm. The values from the least significant bit up to the most significant bit are computed sequentially (i.e., from layer $2^0$ to $2^3$). The computation of each layer ends up with the remaining number of this layer as one. In each computation, a 1-bit adder is used for every three elements. If the number of elements left is two, it is computed together with an auxiliary bit 0.}
	\label{fig:eg_ac}
\end{figure}

\subsection{Theoretical Complexity Analysis}\label{5.3}
In this section, we analyze the number of non-XOR gates our LBA algorithm needs. 

\begin{theorem}\label{theo:lba}
    The number of non-XOR gates LBA algorithm needs, denoted by ${S}_{N}^{lba}$, satisfies $N- \lceil \log_{2}(N+1) \rceil \leq {S}_{N}^{lba} \leq N$.
\end{theorem}

\begin{proof}
    The $1$-bit adder circuit depicted in Figure \ref{fig:adder} is utilized in LBA algorithm. Each circuit inputs three bits and outputs two bits, achieving a one-bit reduction per execution. The input-output bit counts difference dictates the number of non-XOR gates required in LBA algorithm, with each circuit execution needing one such gate.

    Since the number of input bits is $N$, and that of output bits is $\lceil \log_{2}(N+1) \rceil$, LBA algorithm needs at least $N- \lceil \log_{2}(N+1) \rceil$ non-XOR gates, in the case that the computation process of all output bits does not involve any auxiliary bit. Thus, we hold ${S}_{N}^{lba} \geq N- \lceil \log_{2}(N+1) \rceil$. 

    On the other hand, if each output bit requires an auxiliary bit for computation, LBA algorithm needs the most non-XOR gates. In this case, the number of auxiliary bits is $\lceil \log_{2}(N+1) \rceil$, and the number of extra non-XOR gates needed is also $\lceil \log_{2}(N+1) \rceil$. Then, the total number of non-XOR gates is $N- \lceil \log_{2}(N+1) \rceil + \lceil \log_{2}(N+1) \rceil = N$. Therefore, we have ${S}_{N}^{lba} \leq N$.
\end{proof}

As a result, we can conclude that our LBA algorithm outperforms the tree adder by about $(2N-N)/N = 100\%$.

\section{BNN Structure Optimization} \label{sec:BNSCO}
In this section, we further optimize BNN structures to get more efficient boolean circuits based on the BNN customization process described in Section \ref{sec:Sqni}. We use link optimization approach to optimize link connections given a network structure, and structure exploration to explore optimal network structures under same secure computation costs.

\subsection{Link Optimization} \label{sec:LO}
We propose a link optimization method to optimize link connections by first training Ternary Weight-based BNNs (TWBNNs) and then deleting links with zero-weights. Our main idea is as follows. We first quantize full-precision weights to ternary values $\{+\alpha,0,-\alpha\}$, which preserves more network information than binarization due to more values quantized. TWBNNs enhance model information capacity by using TWNs instead of BNNs during training, and thus further mitigating accuracy loss via KD. Then, we train TWBNNs via KD, with original FPWNs as teacher models and TWBNNs as student models. Finally, we prune zero-weight links from trained TWBNNs, yielding Link Reduction-based BNNs (LRBNNs) with optimized binarized ($\{+\alpha,-\alpha\}$) link connections. Pruning zero-weight links in TWBNNs reduces model connections, accelerating inference computation. Figure \ref{fig:link-optimization} depicts the link optimization workflow. 

\begin{figure}[!htbp] 
	\centering
	\includegraphics[width=0.44\textwidth]{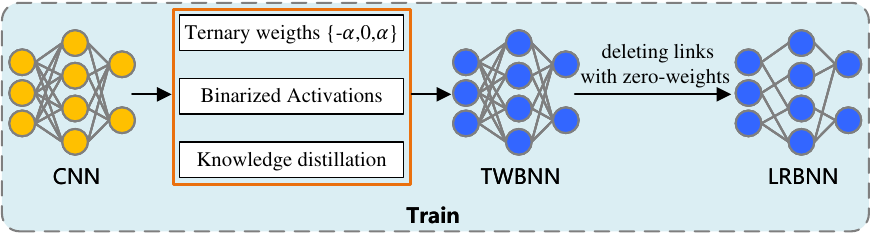}
	\caption{The Workflow of Link Optimization.}
	\label{fig:link-optimization}
\end{figure}

Note that, we can easily convert weights $\{+\alpha,-\alpha\}$ into $\{+1,-1\}$ by passing them through a BN layer of $Y = \frac{1}{\alpha} X$. Via Eq.~\eqref{eq:bn-sign}, we effectively merge linear operations with $Sign$ function and convert them into just an OBC and a comparison. Link optimization yields BNNs retaining only effective non-zero weights. Such BNNs achieve greater efficiency and speed via higher-bit training quantization combined with sparse inference link connections. Therefore, we obtain a more efficient boolean circuit for oblivious inference.

\subsection{Structure Exploration} \label{sec:SE}
We explore BNN structures while maintaining comparable secure computation costs. Figure \ref{fig:equivalent} illustrates a simple modeling of BNNs. The input size and output classes are dataset-constrained, while hidden  layers facilitate classification tasks. Pooling layers partition hidden layers into smaller convolution subsets. The performance costs of BN and BA are negligible compared to convolution, thus excluded from cost analysis. Padding maintains uniform convolution layer dimensions, with only constant-cost subsets considered. Parameter tuning explores cost-constrained optimal BNN structures.

\begin{figure}[!htbp] 
	\centering
	\includegraphics[width=0.42\textwidth]{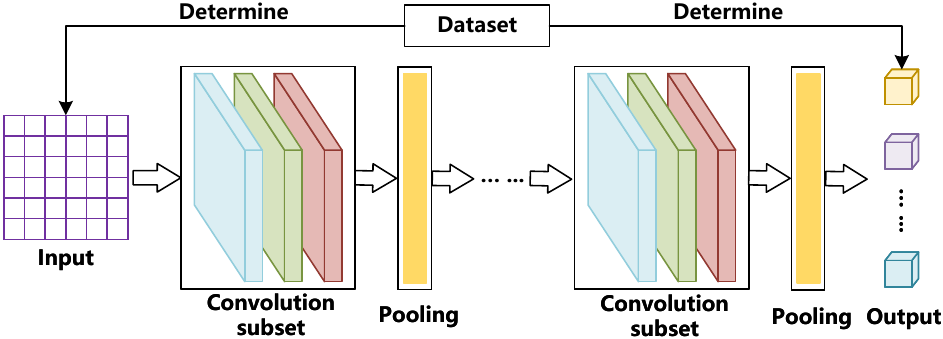}
	\caption{The Modeling of BNNs.}
	\label{fig:equivalent}
\end{figure}

We discuss 1D and 2D convolutions separately. Assume that the network model comprises $z$ convolution subsets, each containing $m$ CONV layers. The filter count and the kernel size in each subset are typically fixed. We compute the cost using the LBA method depicted in Section \ref{5.3}. For ease of analysis, we take ${S}_{N}^{lba} = N$. Thus, the cost of the $m$-th CONV layer in the $z$-th convolution subset is denoted by $S_{N_{m,z}}^{lba}$, where ${S}_{N_{m,z}}^{lba} = N_{m,z}$. 

1) \textbf{1D Convolutions.} The input shape of the $m$-th CONV layer in the $z$-th convolution subset is defined as $h_{m,z}^{1,\operatorname{1D}} \times h_{m,z}^{2,\operatorname{1D}}$. Denote the filter count and the kernel size in the $z$-th convolution subset by $g_{z}^{\operatorname{1D}}$ and $c_{z}$. The output shape is $h_{m,z}^{1,\operatorname{1D}} \times g_{z}^{\operatorname{1D}}$. Based on the above description, we can compute the cost of the $m$-th CONV layer in the $z$-th convolution subset by $C_{m,z}^{\operatorname{1D}} = h_{m,z}^{1,\operatorname{1D}} \cdot g_{z}^{\operatorname{1D}} \cdot N_{m,z}^{\operatorname{1D}} = h_{m,z}^{1,\operatorname{1D}} \cdot g_{z}^{\operatorname{1D}} \cdot c_{z} \cdot h_{m,z}^{2,\operatorname{1D}}$. 
We can get the total cost of the $z$-th convolution subset: $C_{z}^{\operatorname{1D}} =  \sum_{i=1}^{m}C_{i,z}^{\operatorname{1D}}$. 
The cost of all convolution subsets can also be computed: $C^{\operatorname{1D}} =  \sum_{i=1}^{z}C_{z}^{\operatorname{1D}}$.

2) \textbf{2D Convolutions.} The input shape of the $m$-th CONV layer in the $z$-th convolution subset is defined as $h_{m,z}^{1,\operatorname{2D}} \times h_{m,z}^{2,\operatorname{2D}} \times h_{m,z}^{3,\operatorname{2D}}$, where $h_{m,z}^{3,\operatorname{2D}}$ denotes the number of input channels. Denote the filter count and the kernel window size in the $z$-th convolution subset by $g_{z}^{\operatorname{2D}}$ and $o_{z}^{1} \times o_{z}^{2}$. We can know that the output shape is $h_{m,z}^{1,\operatorname{2D}} \times h_{m,z}^{2,\operatorname{2D}} \times g_{z}^{\operatorname{2D}}$. The cost of the $m$-th CONV layer in the $z$-th convolution subset can be computed as follows: $C_{m,z}^{\operatorname{2D}} = h_{m,z}^{1,\operatorname{2D}} \cdot h_{m,z}^{1,\operatorname{2D}} \cdot g_{z}^{\operatorname{2D}} \cdot N_{m,z}^{\operatorname{2D}} = h_{m,z}^{1,\operatorname{2D}} \cdot h_{m,z}^{1,\operatorname{2D}} \cdot g_{z}^{\operatorname{2D}} \cdot h_{m,z}^{3,\operatorname{2D}} \cdot o_{z}^{1} \cdot o_{z}^{2}$.
The total cost of the $z$-th convolution subset can be derived as follows: $C_{z}^{\operatorname{2D}} =  \sum_{i=1}^{m}C_{i,z}^{\operatorname{2D}}$.
We can obtain the total cost of all convolution subsets: $C^{\operatorname{2D}} =  \sum_{i=1}^{z}C_{z}^{\operatorname{2D}}$.

BNN structures are optimized under secure computation cost constraints through parameter exploration. 
Specifically, we always keep $C^{\operatorname{1D}}$ or $C^{\operatorname{2D}}$ constant. There are two ways of modifying parameters: intra-subset and inter-subset. For intra-subset modifications, we can modify the filter count $g_{z}^{\operatorname{1D}}$ (resp. $g_{z}^{\operatorname{2D}}$), the kernel size $c_{z}$ (resp. $o_{z}^{1} \times o_{z}^{2}$), and the number of convolution layers in each subset to keep $C_{z}^{\operatorname{1D}}$ (resp. $C_{z}^{\operatorname{2D}}$) constant. Inter-subset modifications are more complex than intra-subset modifications. We need to balance the costs of several subsets while making intra-subset modifications to preserve global cost constant.

\section{Experimental Evaluation} \label{sec:Experiments}
In this section, we conduct extensive experimental evaluation on FOBNN. First, we evaluate our two OBC algorithms, BLB and LBA. Second, we compare FOBNN with prior art. Third, we deploy FOBNN for real-world RNA function prediction. Finally, we assess two BNN structure optimization approaches, link optimization and structure exploration.

To ensure secure computations, we convert the fundamental operations in each layer to fixed-point arithmetic using GC \cite{evans2018pragmatic,lindell2009proof}. We train our binarized networks with the standard BNN algorithm \cite{courbariaux2016binarized}, implemented in Python with TensorFlow \cite{abadi2016tensorflow} and PyTorch \cite{paszke2019pytorch}. We evaluate our models on various benchmarks using the emp-toolkit library \cite{emp-toolkit} for GC computation. Experiments are implemented on Ubuntu 18.04 with 16GB RAM. We run our benchmarks in two network settings LAN and WAN, where LAN is used by default. The bandwidth is about 10GBps (LAN) and 190MBps (WAN). The latency is about 0.024ms (LAN) and 10ms (WAN), respectively. 
Our key performance metrics are: {\itshape accuracy, runtime,} and {\itshape communication overhead}.

\subsection{Evaluation on OBC Algorithms}
To assess the performance of our OBC algorithms, BLB and LBA, we conduct extensive experiments and compare them with TA, the state-of-the-art OBC algorithm. As shown in Table \ref{tab:cmp}, BLB and LBA achieve up to 1.3$\times$ and 2.5$\times$ faster compared to TA, respectively.
Both schemes significantly reduce non-XOR gate counts and communication overhead, with BLB and LBA decreasing by up to 1.2$\times$ and 2$\times$, respectively. 
The experimental results align closely with our theoretical analysis outlined in Sections \ref{sec:BLB} and \ref{sec:LBA}.

\begin{table}[htbp]
	\caption{Comparison of Runtime (RT.), Non-XOR Gates, and Communication (Comm.) as the Bit Vector Size Varies}
	\label{tab:cmp}
	\centering
    \resizebox{0.42\textwidth}{!}{
		\begin{tabular}{ccccc} 
			\hline 
            \noalign{\vskip 1pt}
			\textbf{Size} & \textbf{Approach} & \textbf{RT. (\textmu s)} & \textbf{Non-XOR gates} & \textbf{Comm. (KB)}\\ 
			\hline
            \noalign{\vskip 1pt}
			\multirow{3}{*}{250} & TA & 59 & 492 & 15.38\\
			& BLB & 47 & 412 & 12.84\\
			& LBA & 28 & 244 & 7.63\\ 
			\hline 
            \noalign{\vskip 1pt}
			\multirow{3}{*}{500} & TA & 110 & 992 & 30.97\\
			& BLB & 90 & 832 & 25.91\\
			& LBA & 48 & 496 & 15.44\\
			\hline
            \noalign{\vskip 1pt}
			\multirow{3}{*}{1000} & TA & 224 & 1992 & 62.19\\
			& BLB & 172 & 1664 & 51.97\\
			& LBA & 89 & 996 & 31.06\\
			\hline
            \noalign{\vskip 1pt}
			\multirow{3}{*}{2000} & TA & 446 & 3992 & 124.66\\
			& BLB & 351 & 3340 & 104.31\\
			& LBA & 185 & 1996 & 62.31\\
			\hline
		\end{tabular}
    }
\end{table}

We implement a two-party SS-based TA (SS-TA) solution to OBC problem built on ABY \cite{demmler2015aby} in LAN and WAN and compare it with our GC-based LBA (GC-LBA) algorithm with same parameters. The SS-TA solution consists of an online phase for computations and an offline phase for arithmetic multiplication triples generation using OT. Table \ref{tab:cmp2ss} shows the comparison results. Our GC-LBA requires only constant rounds of interactions. In contrast, the number of interactions in the two-party SS-TA solution is proportional to the depth of the underlying circuit. The round of interactions dominates communication time. In WAN, the runtime of the two-party SS-TA solution is occupied by its communication time. The BaseOT of SS-TA solution dominates the runtime in LAN and WAN. Therefore, as can be seen in Table \ref{tab:cmp2ss}, our GC-LBA is much more efficient regarding runtime than the SS-TA scheme. The SS-TA solution includes an offline phase, our GC-LBA is slightly better than it in terms of communication overhead.

\begin{table}[!htbp]
	\caption{Performance Comparison between GC-LBA and SS-TA with Various Windows on Convolution Inputs (Cin.)}
	\label{tab:cmp2ss}
	\centering
    \resizebox{0.42\textwidth}{!}{
		\begin{tabular}{cccccc} 
			\hline 
            \noalign{\vskip 1pt}
			\multirow{2}{*}{\textbf{Cin.}} & \multirow{2}{*}{\textbf{Windows}} & \multirow{2}{*}{\textbf{Approach}} & \multicolumn{2}{c}{\textbf{RT. (ms)}} & \multirow{2}{*}{\textbf{Comm. (KB)}}\\ \cline{4-5} 
            \noalign{\vskip 1pt}  
            &  &  & LAN & WAN &  \\
			\hline
            \noalign{\vskip 1pt}
			\multirow{6}{*}{20$\times$20} & \multirow{2}{*}{5$\times$5} & SS-TA & 149.42 & 239.90 & 377.32 \\
			& & GC-LBA & 1.81 & 41.20 & 275 \\ \cline{2-6}
            \noalign{\vskip 1pt} 
            & \multirow{2}{*}{4$\times$4} & SS-TA & 143.659 & 213.57 & 276.66 \\
            & & GC-LBA & 1.35 & 31.7 & 187.5 \\ \cline{2-6}
            \noalign{\vskip 1pt}  
            & \multirow{2}{*}{3$\times$3} & SS-TA & 137.02 & 185.13 & 201.46 \\
            & & GC-LBA & 0.98 & 20.21 & 187.5 \\ 
			\hline
            \noalign{\vskip 1pt}
            \multirow{6}{*}{28$\times$28} & \multirow{2}{*}{5$\times$5} & SS-TA & 158.22 & 270.28 & 643.99 \\
            & & GC-LBA & 3.38 & 52.24 & 539 \\ \cline{2-6}
            \noalign{\vskip 1pt}
            & \multirow{2}{*}{4$\times$4} & SS-TA & 148.93 & 238.17 & 448.41 \\
            & & GC-LBA & 2.81 & 43.43 & 367.5 \\ \cline{2-6}
            \noalign{\vskip 1pt}
            & \multirow{2}{*}{3$\times$3} & SS-TA & 144.61 & 207.25 & 299.40 \\
            & & GC-LBA & 1.74 & 30.21 & 171.5 \\ 
			\hline
		\end{tabular}
    }
\end{table}

\subsection{Comparison with Prior Art}\label{sec:Eva_mnist}
We evaluate two networks (MnistNet1 and MnistNet2) on MNIST dataset and three networks (CifarNet1, CifarNet2, and CifarNet3) on CIFAR-10 dataset. The architectures of MnistNet1, MnistNet2, CifarNet1, CifarNet2, and CifarNet3 are the same as BM2, BM3, BC2, BC1 and BC3 in XONN \cite{riazi2019xonn}. We scale BNN structures with various scaling factor $s$. We choose LeNet and VGG16 with full precision counterparts as teacher models of KD in MNIST and CIFAR-10 for model training, respectively. Table \ref{tab:cmp_xonn} compares FOBNN with previous work in terms of accuracy, runtime and communication. 

\begin{table}[htbp]
	\caption{Comparison of FOBNN with the State-of-the-art for Different Network Architectures (Arch.)}
	\label{tab:cmp_xonn}
	\centering
    \resizebox{0.44\textwidth}{!}{
        \begin{threeparttable}
		\begin{tabular}{cccccc} 
			\hline 
            \noalign{\vskip 1pt} 
			\textbf{Arch.} & \textbf{Framework} & \textbf{RT. (s)} & \textbf{Comm. (MB)} & \textbf{Acc. (\%)} & \textbf{s} \\
			\hline
            \noalign{\vskip 1pt}
			\multirow{10}{*}{MnistNet1} & CryptoNets & 297.5 & 372.2 & 98.95 & - \\
			& DeepSecure & 9.67 & 791 & 98.95 & - \\
			& MiniONN & 1.28 & 47.6 & 98.95 & - \\
            & Chameleon & 2.24 & 10.5 & 99.0 & - \\
            & EzPC & 0.6 & 70 & 99.0 & - \\
            & Gazelle & 0.29 & 8.0 & 99.0 & - \\
            & XONN & 0.29* & 72.53* & 98.64 & 4 \\
            & FOBNN(BLB) & 0.23 & 60.75 & 98.77 & 4 \\
            & FOBNN(LBA) & 0.10 & 37.03 & 98.77 & 4 \\
            & FOBNN(LRBNN) & 0.07 & 23.12 & 98.86 & 4 \\
			\hline 
            \noalign{\vskip 1pt} 
            \multirow{8}{*}{MnistNet2} & MiniONN & 9.32 & 657.5 & 99.0 & - \\
            & EzPC & 5.1 & 501 & 99.0 & - \\
            & Gazelle & 1.16 & 70 & 99.0 & - \\
            & XONN & 0.45* & 118.89* & 99.0 & 2 \\
            & Leia & 37.4 & 328.1 & 99.1 & - \\
            & FOBNN(BLB) & 0.34 & 100.34 & 99.0 & 2 \\
            & FOBNN(LBA) & 0.13 & 64.7 & 99.0 & 2 \\
            & FOBNN(LRBNN) & 0.1 & 43.07 & 99.2 & 2 \\
            \hline
            \noalign{\vskip 1pt} 
            \multirow{9}{*}{CifarNet1} & MiniONN & 544 & 9272 & 82 & - \\
            & Chameleon & 52.67 & 2650 & 82 & - \\
            & EzPC & 265.6 & 40683 & 82 & - \\
            & Gazelle & 15.48 & 1236 & 82 & - \\ 
            & XONN & 22.37* & 6133* & 82 & 3 \\
            & Leia & 123.1 & 919.4 & 72 & - \\
            & FOBNN(BLB) & 17.24 & 5089 & 82 & 3 \\
            & FOBNN(LBA) & 6.12 & 3104 & 82 & 3 \\
            & FOBNN(LRBNN) & 4.2 & 1936 & 84 & 3 \\
            \hline
            \noalign{\vskip 1pt}
            \multirow{4}{*}{CifarNet2} & XONN & 9.35* & 2422* & 72 & - \\
            & FOBNN(BLB) & 6.88 & 1998 & 72 & - \\
            & FOBNN(LBA) & 2.36 & 1219 & 72 & - \\
            & FOBNN(LRBNN) & 1.59 & 746 & 75 & - \\
            \hline
            \noalign{\vskip 1pt} 
            \multirow{5}{*}{CifarNet3} & XONN & 37.75* & 9581* & 83 & 2 \\
            & Leia & 199 & 1829.9 & 81 & - \\
            & FOBNN(BLB) & 29.04 & 7972 & 83 & 2 \\
            & FOBNN(LBA) & 9.53 & 4842 & 83 & 2 \\
            & FOBNN(LRBNN) & 6.14 & 2983 & 85 & 2 \\
            \hline
		\end{tabular}
        \begin{tablenotes}[para,flushleft]
            \footnotesize 
            \item[*] XONN's reported security overhead remains unreproduced due to inaccessible source code and ambiguous implementation details. We report the results under identical FOBNN settings except that XONN uses TA for OBC problems. 
        \end{tablenotes}
        \end{threeparttable}
    }
\end{table}

For evaluation on MNIST, our BLB solution achieves 1293$\times$, 42$\times$, 9.7$\times$, and 110$\times$ faster inference compared to CryptoNets \cite{gilad2016cryptonets}, DeepSecure \cite{rouhani2018deepsecure}, Chameleon \cite{riazi2018chameleon}, and Leia \cite{liu2022leia}, respectively, and our LBA solution is 2975$\times$, 96.7$\times$, 22.4$\times$, and 287.7$\times$ faster. 
Compared to MiniONN \cite{liu2017oblivious}, EzPC \cite{chandran2019ezpc}, Gazelle \cite{juvekar2018gazelle}, and XONN \cite{riazi2019xonn} respectively, BLB achieves up to 27.4$\times$, 15$\times$, 3.4$\times$, and 1.3$\times$ better latency, and LBA has up to 71.7$\times$, 39.2$\times$, 8.9$\times$, and 3.5$\times$ lower latency. 
Table \ref{tab:cmp_xonn} also shows evaluation on CIFAR-10. As can be seen, our BLB solution achieves 31.5$\times$, 3.1$\times$, and 15.4$\times$ faster inference compared to MiniONN, Chameleon, and EzPC respectively, and our LBA solution is 88.9$\times$, 8.6$\times$, and 43.4$\times$ faster. Compared to Gazelle, LBA is 2.5$\times$ faster. Compared to XONN and Leia, BLB is 1.4$\times$ and 7$\times$ faster, and LBA achieves up to 4$\times$ and 20$\times$ faster, respectively. Our BLB and LBA reduce communication overhead by up to 1.2$\times$ and 2$\times$ compared to XONN respectively, which align with our theoretical analyses in Sections \ref{sec:BLB} and \ref{sec:LBA}.

\subsection{Real Application in Bioinformatics}\label{sec:Eva_nRC}
We evaluate FOBNN for RNA function prediction on two real-world bioinformatics datasets: nRC \cite{rossi2019ncrna} and Rfam \cite{kalvari2018rfam}. The nRC dataset contains 6,160 training and 2,529 testing ncRNA sequences across 13 functional classes. The Rfam dataset, featuring 88 distinct classes, comprises 105,864 training, 17,324 validation, and 25,342 testing sequences. We explore 1D k-mer encodings (3mer, 2mer, 1mer) and 2D space-filling curves encodings (Snake, Morton, Hilbert) to convert the raw RNA sequences to binary vectors incorporating the latest padding criteria, and more details can be found in \cite{noviello2020deep}. 
The BNN architecture customized for RNA function prediction initially encodes raw RNA sequences into a 1D or 2D input representation. This architecture then progresses through five CONV layers, each followed by a BN, a BA function, and, if required, a MP layer. Finally, a FC layer with a $Sign$ activation function is added to compress the data, leading to a softmax multi-class classification output layer. 
To accommodate the nRC dataset scale, we configure pooling windows of size 4 in second and fourth layers of the binarized network for 1D k-mer encoding.

We make a performance comparison between the ncrna-deep (i.e., unquantized CNN) model \cite{noviello2020deep} and our customized binarized neural network (CBN) model. 
The comparison results are summarized in Table \ref{tab:Comm_nRC}.
As can be seen, in terms of nRC database, our solution achieves higher accuracy compared to ncrna-deep, and the accuracy of 2D increases more than that of 1D. Our CBN is superior to ncrna-deep, except for 3-mer and 2-mer encodings in the Rfam dataset. These indicate that the accuracy of our CBN can match or even exceed that of the initial CNN, proving the binarization effectiveness of FOBNN. We then perform secure inference using LBA based on our CBN. To analyze the runtime and the communication overhead, we use fundamental circuit operations and the LBA method to implement secure inference of ncrna-deep and our CBN in GC, respectively. 
We can observe from Table \ref{tab:Comm_nRC} that our CBNs perform better than full-precision CNNs in terms of secure computation costs in oblivious inference. It implies fast and efficient computation of LBA solution in FOBNN.

\begin{table}[htbp]
	\caption{Performance Comparison of Oblivious Inference for the nRC and Rfam Datasets with Different Encoding (Enc.)}
	\label{tab:Comm_nRC}
	\centering
    \resizebox{0.46\textwidth}{!}{
		\begin{tabular}{cccccccc} 
			\hline 
            \noalign{\vskip 1pt}
			\multirow{2.3}{*}{\textbf{Enc.}} & \multirow{2.3}{*}{\textbf{Framework}} & \multicolumn{2}{c}{\textbf{RT. (s)}} & \multicolumn{2}{c}{\textbf{Comm. (MB)}} & \multicolumn{2}{c}{\textbf{Acc. (\%)}} \\ \cmidrule(lr){3-4} \cmidrule(lr){5-6} \cmidrule(lr){7-8}  
			&  & nRC & Rfam & nRC & Rfam & nRC & Rfam \\
			\hline
            \noalign{\vskip 1pt}
			\multirow{2}{*}{3mer} & secure ncrna-deep & 47.42 & 12.20 & 37281 & 9868 & 81.61 & 86.43 \\
			& \textbf{secure CBN} & \textbf{10.25} & \textbf{2.95} & \textbf{4290} & \textbf{1303} & \textbf{82.96} & \textbf{84.93} \\ \cline{1-8}
            \noalign{\vskip 1pt}
			\multirow{2}{*}{2mer} & secure ncrna-deep & 51.25 & 13.58 & 41196 & 11010 & 85.01 & 87.63 \\
			& \textbf{secure CBN} & \textbf{7.20} & \textbf{2.35} & \textbf{3155} & \textbf{1141} & \textbf{85.65} & \textbf{86.88} \\ \cline{1-8}
            \noalign{\vskip 1pt}
			\multirow{2}{*}{1mer} & secure ncrna-deep & 93.75 & 24.70 & 75051 & 19884 & 85.69 & 89.24 \\
			& \textbf{secure CBN} & \textbf{9.93} & \textbf{3.96} & \textbf{4770} & \textbf{1922} & \textbf{87.66} & \textbf{90.14} \\
			\hline
            \noalign{\vskip 1pt}
			\multirow{2}{*}{Snake} & secure ncrna-deep & 86.34 & 19.01 & 70479 & 16166 & 78.88 & 81.04 \\
			& \textbf{secure CBN} & \textbf{22.05} & \textbf{5.23} & \textbf{9735} & \textbf{2426} & \textbf{88.06} & \textbf{81.19} \\ \cline{1-8}
            \noalign{\vskip 1pt}
			\multirow{2}{*}{Morton} & secure ncrna-deep & 113.44 & 28.20 & 92624 & 23511 & 79.52 & 81.12 \\
			& \textbf{secure CBN} & \textbf{28.08} & \textbf{6.38} & \textbf{12851} & \textbf{3072} & \textbf{87.35} & \textbf{81.58} \\ \cline{1-8}
            \noalign{\vskip 1pt}
			\multirow{2}{*}{Hilbert} & secure ncrna-deep & 114.85 & 27.67 & 92624 & 23511 & 80.27 & 80.21 \\
			& \textbf{secure CBN} & \textbf{27.49} & \textbf{6.35} & \textbf{12851} & \textbf{3072} & \textbf{86.00} & \textbf{81.62} \\ 
			\hline 
		\end{tabular}
    }
\end{table}

The BLB and LBA are secure computation solutions for the OBC problem. They compute on the same bit vector as the tree-adder structure, and the output is an exact value, so they have no impact on network accuracy. Since we use a GC-based implementation, our solutions only require a constant round of interactions. To ensure fairness, we implement an alternative GC-based approach which adopts TA solution and also requires a constant round of interactions, on the same hardware, and compare our solutions with it.
Table \ref{tab:time_and_nonxor} compares oblivious inference performance of our FOBNN using BLB and LBA solutions with the TA solution.  
It is worth noting that the performance evaluation of oblivious inference involves the whole FOBNN, of which our BLB and LBA algorithms are only a part.
Thus, for overall performance, the non-XOR gates consumption of the LBA architecture is not always half that of the TA architecture.

\begin{table}[htbp]
    \setlength\tabcolsep{3pt} 
	\caption{Total Runtime (TRT.), Total Communication Overhead (TCO.), and Number of Non-XOR Gates (GNum.) for Oblivious Inference of BLB and LBA Compared to TA on nRC and Rfam}
	\label{tab:time_and_nonxor}
	\centering
    \resizebox{0.46\textwidth}{!}{
		\begin{tabular}{ccccccccc} 
			\hline 
            \noalign{\vskip 1pt}
			\multirow{2.2}{*}{\textbf{Dim.}} & \multirow{2.2}{*}{\textbf{Enc.}} & \multirow{2.2}{*}{\textbf{Arch.}} & \multicolumn{2}{c}{\textbf{TRT. (s)}} & \multicolumn{2}{c}{\textbf{TCO. (MB)}} & \multicolumn{2}{c}{\textbf{GNum}} \\ \cmidrule(lr){4-5} \cmidrule(lr){6-7} \cmidrule(lr){8-9} 
			& & & nRC & Rfam & nRC & Rfam & nRC & Rfam \\
			\hline
            \noalign{\vskip 1pt}
			\multirow{9}{*}{1D} & \multirow{3}{*}{3mer} & TA & 18.94 & 5.99 & 5730 & 1788 & 187761396 & 58592164 \\
			& & FOBNN(BLB) & \textbf{15.44} & \textbf{4.84} & \textbf{5103} & \textbf{1593} & \textbf{167586356} & \textbf{52192612} \\ 
            & & FOBNN(LBA) & \textbf{10.25} & \textbf{2.95} & \textbf{4290} & \textbf{1303} & \textbf{140581256} & \textbf{42696116} \\ \cline{2-9}
            \noalign{\vskip 1pt}
			& \multirow{3}{*}{2mer} & TA & 18.74 & 7.08 & 5140 & 1886 & 168429948 & 61808552 \\ 
            & & FOBNN(BLB) & \textbf{14.75} & \textbf{5.22} & \textbf{4435} & \textbf{1628} & \textbf{145328012} & \textbf{53361848} \\ 
            & & FOBNN(LBA) & \textbf{7.20} & \textbf{2.35} & \textbf{3155} & \textbf{1141} & \textbf{103367556} & \textbf{37389492} \\ \cline{2-9}
            \noalign{\vskip 1pt}
            & \multirow{3}{*}{1mer} & TA & 33.81 & 13.41 & 8758 & 3538 & 286969472 & 115919404 \\ 
            & & FOBNN(BLB) & \textbf{25.46} & \textbf{10.53} & \textbf{7427} & \textbf{3003} & \textbf{243371292} & \textbf{98407752} \\ 
            & & FOBNN(LBA) & \textbf{9.93} & \textbf{3.96} & \textbf{4770} & \textbf{1922} & \textbf{156297100} & \textbf{62978744} \\ 
            \hline
            \noalign{\vskip 1pt}
            \multirow{9}{*}{2D} & \multirow{3}{*}{Snake} & TA & 76.55 & 17.09 & 18654 & 4540 & 611257096 & 148762808 \\
            & & FOBNN(BLB) & \textbf{57.56} & \textbf{12.93} & \textbf{15701} & \textbf{3839} & \textbf{514492844} & \textbf{125784984} \\ 
            & & FOBNN(LBA) & \textbf{22.05} & \textbf{5.23} & \textbf{9735} & \textbf{2426} & \textbf{319011984} & \textbf{79490992} \\ \cline{2-9}
            \noalign{\vskip 1pt}
            & \multirow{3}{*}{Morton} & TA & 99.44 & 22.06 & 24635 & 5794 & 807247744 & 189841972 \\ 
            & & FOBNN(BLB) & \textbf{76.12} & \textbf{16.88} & \textbf{20734} & \textbf{4892} & \textbf{679402280} & \textbf{160306976} \\ 
            & & FOBNN(LBA) & \textbf{28.08} & \textbf{6.38} & \textbf{12851} & \textbf{3072} & \textbf{421116312} & \textbf{100651068} \\ \cline{2-9}
            \noalign{\vskip 1pt}
            & \multirow{3}{*}{Hilbert} & TA & 100.35 & 22.07 & 24635 & 5794 & 807247744 & 189841972 \\ 
            & & FOBNN(BLB) & \textbf{75.74} & \textbf{17.06} & \textbf{20734} & \textbf{4892} & \textbf{679402280} & \textbf{160306976} \\ 
            & & FOBNN(LBA) & \textbf{27.49} & \textbf{6.35} & \textbf{12851} & \textbf{3072} & \textbf{421116312} & \textbf{100651068} \\ 
			\hline
		\end{tabular}
    }
\end{table}

In summary, compared to TA in the nRC database, our BLB solution reduces the total runtime by up to 1.3$\times$ in 1D $k$-mer and 2D space-filling curves encodings. Our LBA solution achieves up to 3.4$\times$ and 3.6$\times$ faster inference in 1D $k$-mer and 2D space-filling curves encodings respectively. As can be seen, BLB and LBA reduce the communication overhead by up to 1.2$\times$ and 1.8$\times$ in 1D $k$-mer encodings, and 1.2$\times$ and 1.9$\times$ in 2D space-filling curves encodings, respectively. 
As can be seen in the Rfam database, in terms of the total runtime, our BLB solution is up to 1.4$\times$ and 1.3$\times$ faster than TA in 1D $k$-mer and 2D space-filling curves encodings respectively, and LBA achieves up to 3.4$\times$ and 3.5$\times$ better latency compared to TA respectively. We can observe that BLB and LBA reduce the communication overhead by up to 1.2$\times$ and 1.8$\times$ in 1D $k$-mer encodings, and 1.2$\times$ and 1.9$\times$ in 2D space-filling curves encodings, respectively. 
The number of non-XOR gates consumed is closely related to the communication overhead and thus has the same performance. 

The evaluation results for communication overhead, non-XOR gate consumption, and runtime performance align with our theoretical analyses in Sections \ref{sec:BLB} and \ref{sec:LBA}. 
LBA outperforms theoretical predictions via its efficient single-bit computation scheme, surpassing TA and BLB's multi-bit schemes on EMP-toolkit.

\subsection{BNN Structure Optimization} \label{sec:QNSO}
\subsubsection{Link Optimization} An accuracy gap exists between MnistNet1, MnistNet2, CifarNet1, CifarNet2, CifarNet3 and their FPWNs counterparts. Therefore, we substitute binary weights in BNNs with ternary weights and delete links with zero-weights depicted in Section \ref{sec:LO} to form LRBNNs. We retrain the models with KD and perform secure inference. Table \ref{tab:cmp_xonn} presents LRBNNs' performance. In summary, LRBNNs have better network accuracy than BNNs due to the enhanced ability of weight representation (i.e., binary weights increased to ternary weights). Compared to BNNs with LBA solution, LRBNNs achieve 1.4$\sim$1.6$\times$ faster inference, and reduce communication by up to 1.6$\times$. This decline arises because zero-weights are not involved in oblivious inference computation.

\subsubsection{Structure Exploration} Our CBNs achieve superior accuracy compared to vanilla CNNs in RNA function prediction. Moreover, the model's binary inputs (i.e., 0 or 1) are incompatible with link optimization involving zero-weights. Thus, we apply equivalent modifications in Section \ref{sec:SE} to restructure three networks per CBN. We adopt intra-subset modifications. For 1D convolutions, three modified architectures are defined: MN1 (kernel size reduced to 5 from 10), MN2 (enlarged to 20), and MN3 (constant kernel size with increased CONV layers). In 2D convolutions, kernel size is reduced to 2$\times$2 in MN1 and increased to 4$\times$4 in MN2, while MN3 maintains a 3$\times$3 kernel with increased CONV layers. To maintain constant oblivious inference cost, kernel count scales with kernel size.

\begin{table*}[!htbp]
	\caption{Performance Comparison of Various Modified Architectures while Maintaining Constant Costs in Oblivious Inference}
	\label{tab:Acc_cmp}
	\centering
    \resizebox{0.80\textwidth}{!}{
		\begin{tabular}{ccccccccccccccc} 
			\hline 
            \noalign{\vskip 1pt}
			\multirow{2}{*}{\textbf{Dataset}} & \multirow{2}{*}{\textbf{Encoding}} & \multirow{2}{*}{\textbf{Input}} & \multicolumn{4}{c}{\textbf{Accuracy (\%)}} & \multicolumn{4}{c}{\textbf{Runtime (s)}} & \multicolumn{4}{c}{\textbf{Comm. (MB)}} \\ \cmidrule(lr){4-7} \cmidrule(lr){8-11} \cmidrule(lr){12-15} 
		    & & & \textbf{CBN} & \textbf{MN1} & \textbf{MN2} & \textbf{MN3} & \textbf{CBN} & \textbf{MN1} & \textbf{MN2} & \textbf{MN3} & \textbf{CBN} & \textbf{MN1} & \textbf{MN2} & \textbf{MN3} \\
			\hline
            \noalign{\vskip 1pt}
			 \multirow{3}{*}{nRC} & 3mer & 257$\times$65 & \textbf{82.96} & 82.44 & 81.34 & 80.43 & 10.25 & 10.59 & 9.79 & 10.31 & 4290 & 4364 & 4215 & 4299 \\
			 & 2mer & 385$\times$17 & 85.65 & \textbf{86.00} & 80.55 & 84.34 & 7.20 & 6.92 & 6.56 & 7.04 & 3155 & 3243 & 3066 & 3160 \\
			 & 1mer & 770$\times$5 & 87.66 & \textbf{88.22} & 83.04 & 87.54 & 9.93 & 9.90 & 8.81 & 10.19 & 4770 & 4890 & 4677 & 4805 \\ 
             \hline
             \noalign{\vskip 1pt}
			 \multirow{3}{*}{Rfam} & 3mer & 67$\times$65 & \textbf{84.93} & 84.80 & 82.22 & 83.60 & 2.95 & 3.11 & 2.88 & 3.04 & 1303 & 1366 & 1229 & 1306 \\
			 & 2mer & 100$\times$17 & 86.88 & \textbf{88.11} & 85.54 & 87.00 & 2.35 & 2.49 & 2.28 & 2.36 & 1141 & 1210 & 1047 & 1146 \\
			 & 1mer & 200$\times$5 & \textbf{90.14} & 89.77 & 88.94 & 89.11 & 3.96 & 4.02 & 3.56 & 4.12 & 1922 & 2001 & 1815 & 1933 \\
			\hline
            \noalign{\vskip 1pt}
			 \multirow{3}{*}{nRC} & Snake & 28$\times$28$\times$5 & \textbf{88.06} & 83.79 & 84.86 & 86.71 & 22.05 & 18.76 & 18.42 & 21.49 & 9735 & 9213 & 8827 & 9802 \\
			 & Morton & 32$\times$32$\times$5 & \textbf{87.35} & 80.78 & 83.99 & 84.97 & 28.08 & 25.28 & 24.27 & 28.42 & 12851 & 12136 & 11696 & 12939 \\
			 & Hilbert & 32$\times$32$\times$5 & 86.00 & 82.09 & 83.12 & \textbf{86.12} & 27.49 & 25.27 & 23.97 & 28.74 & 12851 & 12136 & 11696 & 12939 \\
             \hline
             \noalign{\vskip 1pt}
			 \multirow{3}{*}{Rfam} & Snake & 15$\times$15$\times$5 & 81.19 & 83.90 & 82.85 & \textbf{84.24} & 5.23 & 4.78 & 4.53 & 5.27 & 2426 & 2337 & 2168 & 2443 \\
			 & Morton & 16$\times$16$\times$5 & \textbf{81.58} & 78.83 & 79.30 & 81.28 & 6.38 & 6.27 & 5.91 & 6.57 & 3072 & 2950 & 2733 & 3093 \\
			 & Hilbert & 16$\times$16$\times$5 & 81.62 & 80.95 & 81.40 & \textbf{82.21} & 6.35 & 6.17 & 5.96 & 6.49 & 3072 & 2950 & 2733 & 3093 \\
			\hline
		\end{tabular}
    }
\end{table*}

To evaluate comparative performance under constant costs, we retrained modified networks and conducted secure inference experiments. As shown in Table \ref{tab:Acc_cmp}, for large second input dimensions in 1D convolutions, CBN demonstrates superior accuracy, closely matched by MN1. It indicates maintaining CNN-like kernel sizes while augmenting CONV layers can enhance secure inference. However, for substantial dimensional disparities (e.g., 2mer and 1mer), MN1 outperforms others, suggesting kernel size reduction and network depth augmentation are crucial for optimal binarized network design. Conversely, in cases like the 1mer encoding of the Rfam dataset, where the second dimension is small, CBN excels, highlighting the need for tailored binarized network architectures.

Table \ref{tab:Acc_cmp} compares 2D convolution results. CBN and MN3 consistently outperform other architectures, indicating that retaining the vanilla CNN kernel size optimizes 2D convolution binarization efficiency. Larger input pixels necessitate minimal CONV layer increments for optimal performance. Reduced input pixels require increased network depth for accuracy enhancement. Thus, optimal secure inference accuracy requires retaining the vanilla CNN kernel size and elaborately designing the binarized network structure. 

Table \ref{tab:Acc_cmp} shows negligible runtime and communication cost variations across four architectures, satisfying constant-cost oblivious inference requirements. Therefore, optimal network structure selection requires balancing accuracy-runtime-communication tradeoff tailored to application-specific constraints. 

\section{Related Work} \label{sec:RW}
Secure neural network inference has garnered considerable attention in recent years. CryptoNets \cite{gilad2016cryptonets} offers a practical approach to oblivious inference using HE. By integrating cryptography, machine learning, and engineering, it achieves secure inference while maintaining high performance and accuracy. SecureML \cite{mohassel2017secureml} introduces efficient protocols for privacy-preserving machine learning, including linear, logistic regression, and neural networks, leveraging GCs and SS. Notably, it realizes the first privacy-preserving system for neural network training. Since then, solutions like ABY3 \cite{mohassel2018aby3}, SecureNN \cite{wagh2019securenn}, Falcon \cite{wagh2021falcon}, CHET \cite{dathathri2019chet}, 
CrypTFlow2 \cite{rathee2020cryptflow2}, and Cheetah \cite{huang2022cheetah} have been proposed, all relying on SS and HE for secure neural network inference.

DeepSecure \cite{rouhani2018deepsecure} offers a secure deep learning framework leveraging Yao's GCs. However, it lacks guarantees on the privacy of network parameters and data structures. Hybrid oblivious inference protocols have emerged, including MiniONN \cite{liu2017oblivious}, Gazelle \cite{juvekar2018gazelle}, and Chameleon \cite{riazi2018chameleon}, which rely on mixed protocols of GCs, SS, and HE. These approaches vary in complexity and are not tailored for deep learning secure computation. XONN \cite{riazi2019xonn} pioneers the use of GC with BNNs for efficient XNOR-based oblivious inference. However, it does not fully exploit GC's computational advantages by optimizing costly linear operations. Our work focuses on optimizing neural networks for oblivious BNN-based inference, converting expensive linear operations into GC-compatible alternatives. Notable recent research on GCs is arithmetic GCs \cite{heath2024efficient}. However, in this paper, traditional neural networks have been properly binarized and converted into very efficient boolean circuits for GCs with a good accuracy preserved, eliminating the need for arithmetic GCs. Additionally, there is still no practical implementation of arithmetic GCs as far as we know.

There exist researches on secure inference for quantum neural networks using SS. Banners \cite{ibarrondo2021banners}, SecureBiNN \cite{zhu2022securebinn} and FlexBNN\cite{dong2023flexbnn} offer various three-party secure computation solutions for BNNs with replicated SS. Zhang \emph{et al.} \cite{zhang2023scalable} explore GC and additive SS for oblivious BNN inference. However, these methods do not achieve constant-round interactions, regardless of computational depth. FSSiBNN \cite{yang2024fssibnn} achieve secure BNN inference based on function SS, but it is not inherently purely two-party secure computation setting.
Our GC-based oblivious inference offers fast computation with constant-round interactions. Our BLB and LBA solutions can serve as primitives for circuit design, potentially enhancing performance when implemented with SS in future work.

\section{Conclusion} \label{sec:Conclusion}
We have introduced FOBNN, a fast Garbled Circuit-based oblivious inference framework via binarized neural networks. Our FOBNN converts traditional CNNs into customized BNNs, and in turn converts linear operations into nearly cost-free XNOR operations and an OBC problem, and non-linear operations mainly to oblivious comparisons, and finally gets a boolean circuit, which is transformed to a garbled circuit to achieve oblivious inference. To address the OBC problem, the newly emerging performance bottleneck for BNN oblivious inference, we have devised two fast algorithms, Bit Length Bounding and Layer-wise Bit Accumulation, which are of independent interest for secure computations. Theoretical complexity analysis reveals that both algorithms outperform the state-of-the-art OBC algorithm by up to 100\%. Furthermore, we have optimized BNN structures by link optimization under a given BNN structure and structure exploration under constraints of same secure computation costs to obtain more efficient boolean circuits. We have implemented the proof-of-concept system and demonstrated that FOBNN outperforms prior art greatly, while maintaining the advantage of constant number of communication rounds. 

Although our FOBNN enables fast oblivious inference, extending it to large language models is a direction for future research. 
Future work aims to scale our solutions to provide efficient oblivious inference for larger models (e.g., transformers) by enhancing the model's complexity, adopting multi-bit quantization, and designing a novel quantized secure computation algorithm.

{\appendices
\section{Derivation Details of \texorpdfstring{$G_{\kappa}$}{}}  \label{app:2}
In each group computation, the inputs are $2^{\kappa}$-bit numbers and the output is a $2^{\kappa+1}$-bit number. The first $2^{2^\kappa}+1$ numbers are computed using an inverted binary tree structure, where the number of $\ell$-bit adders ($\ell \in [2^{\kappa}..2^{\kappa+1}-1]$) is $2^{2^{\kappa+1}-1-\ell}$. Thus, the total number of non-XOR gates for the first $2^{2^\kappa}+1$ numbers can be computed as follows.
\begin{equation}\label{eq6}
    \begin{aligned}
        {S}_{t} =& \sum_{\ell=2^{\kappa}}^{2^{\kappa+1}-1}(2^{2^{\kappa+1}-1-\ell} \cdot \ell) \\
        =& 2^{2^{\kappa}-1} \cdot 2^{\kappa} + 2^{2^{\kappa}-2} \cdot (2^{\kappa}+1) + \cdots + (2^{\kappa+1}-1)
    \end{aligned}
\end{equation}
Multiplying both sides of Eq.~\eqref{eq6} by $\frac{1}{2}$ yields Eq.~\eqref{eq7}.
\begin{equation}\label{eq7}
    \frac{1}{2} {S}_{t} = 2^{2^{\kappa}-2} \cdot 2^{\kappa} + 2^{2^{\kappa}-3} \cdot (2^{\kappa}+1) + \cdots + \frac{1}{2} \cdot (2^{\kappa+1}-1)
\end{equation}
Eq.~\eqref{eq6} minus Eq.~\eqref{eq7}:
\begin{equation}\label{eq16}
    \begin{aligned}
        \frac{1}{2} {S}_{t} &= 2^{2^{\kappa}-1} \cdot 2^{\kappa} + 2^{2^{\kappa}-2} + 2^{2^{\kappa}-3} + \cdots + 1 - \frac{(2^{\kappa+1}-1)}{2}  \\
        {S}_{t} &= 2^{2^{\kappa}} \cdot 2^{\kappa} + (2^{2^{\kappa}-1} + 2^{2^{\kappa}-2} + \cdots + 2) - (2^{\kappa+1}-1) \\
        {S}_{t} &= 2^{2^{\kappa}} \cdot 2^{\kappa} + 2 \cdot \frac{1-2^{2^{\kappa}-1}}{1-2} - 2^{\kappa+1} + 1 \\
        {S}_{t} &= 2^{2^{\kappa}} \cdot 2^{\kappa} + 2^{2^{\kappa}} - 2^{\kappa+1} -1
    \end{aligned}
    \nonumber
\end{equation}
    


The sum of first $2^{2^\kappa}+1$ numbers produces a $2^{\kappa+1}$-bit number. The non-XOR gates generated by summing it with the remaining $2^{\kappa}$-bit number also needs to be computed. Thus, $G_{\kappa} = {S}_{t} + (2^{\kappa+1}-1) = 2^{2^{\kappa}} \cdot 2^{\kappa} + 2^{2^{\kappa}} - 2$. 
}

\bibliographystyle{IEEEtran}
\bibliography{ref}

\end{document}